\documentclass[final, journal, 10pt]{IEEEtran}

\usepackage{amssymb,amsfonts,amsmath,amsthm,amscd,dsfont,mathrsfs}
\usepackage{eqsection,graphicx,float,epsfig,color}
\usepackage{mathtools,hyperref}
\usepackage{cite}
\newtheorem{theorem}{Theorem}[section]
\newtheorem{lemma}[theorem]{Lemma}
\newtheorem{corollary}[theorem]{Corollary}
\newtheorem{definition}{Definition}[section]
\newtheorem{remark}[definition]{Remark}

\newtheorem{proposition}[theorem]{Proposition}
\newenvironment{proofof}[1]{{\bf Proof of #1:}}{$\qed$\par}

\newcommand{\comment}[1]{}

\def\III{|\!|\!|}

\def\prob{{\mathbb P}}
\def\alg{{\sf Alg}}

\def\bigo{{\sf O}}

\newcommand{\reals}{{\mathds R }}

\newcommand{\naturals}{{ \mathds N}}

\def\cE{{\cal E}}
\def\Var{{\rm Var}}
\def\cA{{\cal A}}
\def\E{{\mathbb E}}

\def\alg{{\sf Alg}}

\def\prob{{\mathbb P}}

\def\eps{\epsilon}
\def\psucc{{\rm P}_{\rm succ}}

\def\normal{{\sf N}}
\def\de{{\rm d}}

\def\complex{{\mathbb C}}
\def\cL{{\mathcal L}}
\def\cG{{\cal G}}
\def\cE{{\cal E}}
\def\cV{{\cal V}}
\def\dg{{\rm deg}}
\def\sign{{\rm sign}}
\def\hG{\widehat{G}}
\def\hQ{\widehat{Q}}
\def\hM{\widehat{M}}
\def\tr{{\rm Tr}}
\def\supp{{\rm supp}}
\def\Param{\Theta}
\def\tParam{\tilde{\Theta}}
\def\hParam{\widehat{\Theta}}

\def\param{\theta}
\def\rls{{\sf Rls}}

\def\bF{{\bf F}}

\def\<{\langle}
\def\>{\rangle}

\begin{document}

% "Title of the paper"
\title{Support Recovery for the Drift Coefficient of High-Dimensional Diffusions}
\author{~Jos\'e~Bento, and Morteza~Ibrahimi% <-this % stops a space
\thanks{J. Bento is with the Department of Electrical Engineering, Stanford University, Stanford, CA, 94305 USA email: jbento@stanford.edu.}% <-this % stops a space
\thanks{M. Ibrahimi is with the Department of Electrical Engineering, Stanford University, Stanford,
CA, 94305 USA email: ibrahimi@stanford.edu.}% <-this % stops a space
}
\date{\today}

\maketitle

\begin{abstract}
Consider the problem of learning the drift coefficient of a $p$-dimensional
stochastic differential equation from a sample path of length $T$.
We assume that the drift is parametrized by a
high-dimensional vector,
and study the support recovery problem when both $p$ and $T$ can tend to infinity.
In particular, we prove a general lower bound on the sample-complexity $T$
by using  a characterization of mutual information as a time integral
of conditional variance, due to Kadota, Zakai, and Ziv.
For linear stochastic differential equations, the drift coefficient
is parametrized by a $p\times p$ matrix which describes
which degrees of freedom interact under the dynamics.
In this case, we analyze a $\ell_1$-regularized least squares
estimator and prove an upper bound on $T$ that nearly matches the
lower bound on specific classes of sparse matrices.
\end{abstract}


%
%********************************************************
%
\section{Introduction}
\label{sec:Intro}

Consider a continuous-time stochastic process $\{x(t)\}_{t \geq 0}$,
$x(t) = [ x_1(t), \dots, x_p(t)] \in\reals^p,$ that is defined by a
stochastic differential equation (SDE) of diffusion type
\begin{eqnarray}
\de x(t) = 
F(x(t);\Param^0) \, \de t + \de b(t) \, , \label{eq:BasicModel}
\end{eqnarray}
where $b(t)$ is a $p$-dimensional standard Brownian
motion and the \textit{drift coefficient}\footnote{Throughout the
paper, vectors are `column vector' even if they are represented in
row form for typographical reasons.} $$F(x(t);\Theta^0) =
[F_1(x(t);\Theta^0),\dots,F_p(x(t);\Theta^0)] \in \reals^p,$$ is a function of $x(t)$
parametrized by $\Theta^0$. This  is an unknown vector, with
dimensions scaling polynomially with $p$.

In this paper we consider the problem of learning the support of
the vector  $\Theta^0$ from a sample trajectory $X^T_0 \equiv \{x(t):\; t \in  [0,T]\}$.  
More precisely, we focus  on the high-dimensional scenario
where $p$ and $T$ are allowed to increase simultaneously.	
Our goal is to determine necessary and sufficient conditions for recovering the support of $\Theta^0$ and the sign of its entries with high probability.
We refer to the smallest $T$ that allows to achieve a prescribed
success probability as the `sample-complexity' of the problem
(although the number of samples is, strictly speaking, infinite). We
are particularly interested in achieving the optimal scaling of sample
complexity with the problem dimensions through computationally 
efficient procedures.

Concretely, given a SDE parametrized by $\Param^0$ and an algorithm
$\alg =\alg(X^T_0)$ that outputs an estimate $\hParam$,
we define the sample-complexity $T_{\alg}(\Param^0)$ as
\begin{align}\label{eq:def_sample_complexity}
\inf \big\{ T_0 \in \reals^{+} : \prob_{\Param^0,T} \{\sign(\hParam) =
  \sign(\Param^0) \} \geq 1- \delta, \hfill& \nonumber \\
\text{ for all } T \geq T_0 & \big\}.
\end{align}
In the expression above, $\prob_{\Param^0,T}$ denotes
probability with
respect to the trajectory $X^T_0$. The function $\sign(.)$
acts element-wise on its vector-valued argument and
to each scalar applies the mapping $\sign: \reals \mapsto \{-1, 0, 1\}$ such that 
\begin{equation*}
\sign(x) = \left\{
\begin{array}{l l}
  -1 & \text{if} \quad x < 0,\\
  +1 & \text{if} \quad x > 0,\\
  0  & \text{if} \quad x = 0.
\end{array}
\right.
\end{equation*}
Obviously,
$T_{\alg}(\Param^0)$ defined above is an upper bound for
sample-complexity of learning the support alone. 
In addition to this definition, given some class
$\cA$ of parameters, we define
\begin{equation}\label{eq:def_sample_complexity_group}
T_{\alg}(\cA) = \sup_{\Param^0 \in \cA} T_{\alg}(\Param^0).
\end{equation}

%%
%As an example, gene regulatory networks can be modeled
%by systems of non-linear stochastic differential equations,
%whose variables encode concentrations of certain gene expression 
%products (e.g. proteins) \cite{BioBook}. Complex chemical networks
%are also described by SDEs that can
% involve hundreds of reactants \cite{Gillespie,Higham}.
%The problem of learning the parameters (reaction coefficients) of
%such an SDE or reconstructing the 
%underlying network structure (i.e. which parameters are non-vanishing) plays crucial role in this context \cite{ToniEtAl}.
%%

Models based on SDEs play a crucial role in several domains of
science and technology, ranging from chemistry to
finance. Consequently, estimating their parameters
has been a topic of great interest in several fields. We refer to Section \ref{sec:Related} for a
brief overview.  A complete understanding of support recovery in a high-dimensional setting is nevertheless missing. 

Our results address these challenges for special classes of SDEs of
immediate relevance. 
A first class is constituted by drift coefficients that are
parametrized linearly. Explicitly, we are given a set of basis functions
\begin{equation}
\bF(x) = [f_1(x), f_2(x), \dots, f_m(x)],
\end{equation}
with $f_i:\reals^p \rightarrow \reals$. The drift is then given as $F(x;\Param^0) = \Param^0 \bF(x)$,
with matrix $\Param^0 \equiv \left\{\param^0_{ij}\right\}_{i\in[p],j\in [m]}\in \reals^{p \times
  m}$, $[p] = \{1, ...,p\}$ and $[m] = \{1, ...,m\}$. We then have, for each $i \in \reals^p$,
\begin{equation}
\label{eq:BasicModelBasis}
\de x_i(t) = 
\sum^m_{j = 1} \param^0_{ij} f_j(x(t))\, \de t + \, \de b_i(t)\, . 
\end{equation}
Suitable sets of basis functions can be provided by domain-specific
knowledge. As an example, within stochastic models of chemical
reactions, the drift coefficient is a low-degree polynomial. For instance, the reaction ${\sf A}+2{\sf B}\to{\sf C}$
is modeled  as $\de x_{{\sf C}} = k_{{\sf C,AB}} \, x_{{\sf A}} x_{{\sf
  B}}^2\de t - k_{{\sf AB,C}}\,x_{{\sf C}} + \de b_{{\sf C}}$ where $x_A$, $x_B$ and $x_C$ denote the
concentration of the species $A$, $B$ and $C$ respectively, and $\de
b_C$ is a chemical noise term. In order to learn a model of this type, one
can consider a basis of functions $\bF(x)$ that comprises all monomials 
up to a maximum degree. In this case, the support of $\Param^0$ tells
which species react with which species, i.e. a network interactions. The sign of its entries distinguishes ``inhibitory'' effects from ``excitatory'' effects.
In the end of this section we give a concrete example of using our method to learn chemical reactions.

An important subclass of models of the last type is provided by linear SDEs. 
In this case, the drift is a linear function of $x(t)$, namely 
$F(x;\Param^0) = \Param^0 x(t)$ with $\Param^0 \equiv \left\{\param^0_{ij}\right\}_{i,j\in[p]}\in \reals^{p\times p}$.
Explicitly, for each $i \in \reals^p$,
\begin{eqnarray}
\de x_i(t) = 
\sum^p_{j = 1} \param^0_{ij} x_j(t)\, \de t + \, \de b_i(t) \, . \label{eq:BasicModelLin}
\end{eqnarray}
A model of this type is a good approximation for many systems near a stable equilibrium.
The model \eqref{eq:BasicModelLin} can be used to trace fluctuations
of the species' concentrations  in proximity
of an equilibrium point in chemical reactions.
In this case, the matrix $\Param^0$ would represent the linearized interactions between different chemical factors.

More generally, we can associate to the model (\ref{eq:BasicModelLin})
a directed graph $G= (V,E)$ with edge weight $\theta^{0}_{ij}\in\reals$
associated to the directed edge $(j,i)$ from $j\in V$ to $i\in V$.
Each component $x_i(t)$ of the vector $x(t)$ describes the state of a
node $i\in V$. The graph $G$ describes which nodes interact: the rate
of change of $x_i(t)$ is given by a weighted sum of the current values
of its neighbors, corrupted by white noise.
In other words, linear SDEs can be seen as graphical models -- a probabilistic model parametrized by a graph.

This paper establishes lower bounds on the sample-complexity for
estimating the support of $\Param^0$ in the
general model (\ref{eq:BasicModel}). These are based on
information theoretic techniques and apply irrespective of
computational considerations. For linear models of the form
(\ref{eq:BasicModelLin}), we put forward a low-complexity estimator
and derive upper bounds on its sample-complexity.
Upper and lower bounds are shown to be within a constant factor for
special classes of sparse networks $\Theta^0$.

Before stating our results more formally, it is useful to stress two
key differences with respect to other high-dimensional
estimation problems.
\begin{enumerate}
\item[$(i)$] \emph{Samples are not independent}.
\item[$(ii)$] \emph{Infinitely many samples are
given as data} (in fact a collection indexed by 
$t\in [0,T]$).  
\end{enumerate}
A simple approach would be to select a finite subsample set.
For instance, one can select a sampling interval $\eta>0$ and
only use samples at regularly spaced times
$\{x(\eta),x(2\eta),x(3\eta),\dots\}$. At first sight, this 
reduces the problem to a more classical one. A closer consideration illustrates instead the
new challenges posed by the present model. 
\begin{itemize}
\item If $\eta$ is small, one obtains a large  number of strongly
  dependent samples and earlier analysis does not apply. In particular, a careful analysis
  must reveal that there is limited information to be harnessed from a
  given time interval $T$.
\item  One might be lead into the conclusion that $\eta$ must be taken sufficiently large as to make samples approximately
independent. However, this approach will waste important information
contained in the sample path.
For example, for a linear SDE, the matrix $\Theta^0$ contains more information than the stationary
distribution of the process (\ref{eq:BasicModelLin})\footnote{Let $\Theta^0_1 = \{\{-2,-1,-1\},\{1,-2,-1\},\{1,1,-2\}\}$ and $\Theta^0_2 = \{\{-2,1,0\},\{-1,-2,1\},\{0,-1,-2\}\}$. The linear systems defined by these matrices have different support. Yet, their stationary behavior is described by the same covariance matrix $\Sigma = \{\{1/4,0,0\},\{0,1/4,0\},\{0,0,1/4\}\}$.}.
\end{itemize}
Our results confirm in a detailed and quantitative way these intuitions.

%
%*********************************************************
%
\subsection{Regularized least squares}
\label{sec:ResContinuumRegularized}

Regularized least squares, $\rls$, is an efficient and well-studied
method for support recovery. We discuss relations 
with existing literature in Section \ref{sec:Related}.
In this paper we study its application to estimating the drift
coefficient of a high-dimensional diffusion and show that 
its sample-complexity compares favorably with our
information-theoretic lower bounds.

Its use is better explained for 
the general linearly parametrized model \eqref{eq:BasicModelBasis}.
For this model, we estimate independently each
row of the matrix $\Param^0 \in \reals^{p \times m}$.
The $r^{{\rm th}}$ row, denoted by $\Param^0_r$, is estimated by solving the following 
convex optimization problem for $\Param_r\in\reals^p$
\begin{equation}
{\rm minimize}\;\;\;
 \cL(\Param_r;X^T_0) + \lambda \| \Param_r\|_1\, ,\label{eq:cont_reg_prob}
\end{equation}
where the log-likelihood function $\cL$ is defined by
\begin{align}
\cL(\Param_r;X^T_0) = 
& \frac{1}{2T} \int^T_0\!\! \<\Param_r, \bF(x(t))\>^2\,\, \de t \nonumber \\
& \quad - \frac{1}{T}\int^T_0\!\! \<\Param_r, \bF(x(t))\>\,\, \de x_r(t)\, . \label{eq:ContCost}
\end{align}
Here and below $\<u,v\>$ denotes the standard scalar product of
vectors $u,v\in\reals^N$.

We denote this algorithm by $\rls(\lambda)$.
The $\ell_1$ regularization term in Eq.~(\ref{eq:cont_reg_prob})
has the role of shrinking to $0$ all the entries $\param_{rj}$, except
the most significant ones, thus effectively selecting the support of $\Param$.

By minimizing the function $\cL$ alone, i.e. setting $\lambda = 0$, one obtains the maximum likelihood
estimator for the diffusion process \eqref{eq:BasicModel}.
Maximum likelihood optimization has been used
before in the context of estimating
diffusions in the low-dimension setting
\footnote{Low-dimensional in the sense
of keeping the number of degrees of
freedom, $p$, fixed and letting $T$ converge
to infinity.}. See \cite{brown1975asymptotic}
and other references in Section \ref{sec:Related}.
In particular, the normalized log-likelihood function (\ref{eq:ContCost}) is the
appropriate generalization of the sum of square residuals for a
continuous-time process.
To see this heuristically, one can \emph{formally} write $\dot{x}_r(t)
=\de x_r(t)/\de t$.
A careless sum of square residuals would take the form $\int (\<\Param_r,
\bF(x(t)) \>- \dot{x}_r(t))^2\de t$.
Unfortunately, this expression is not defined because  $x_r(t)$ is not differentiable.
On the other hand, expanding the square, we get
$2T\cL(\Param_r;X^T_0) +\int (\dot{x}_r(t))^2\de t$. The first
term is well defined, as is clear from Eq.~(\ref{eq:ContCost}), and
the second is independent of $\Param$ and hence can be dropped.

Notice that constructing a well-defined cost function as in
Eq.~(\ref{eq:ContCost}) is not a purely academic problem. Indeed, a
cost function that included the time derivative $\dot{x}(t)$ would in
practice require to estimate $\dot{x}(t)$  itself. This is all but hopeless because $\dot{x}(t)$ does not
exist in the model.

%
%**********************************************************
%

%\newpage
\section{Main Results}

Our main contributions are the followings:

\noindent\textbf{Information-theoretic lower bound:} We establish a general lower
  bound on the sample-complexity for estimating the drift coefficient
  of a diffusion of the form (\ref{eq:BasicModel}). By specializing
  this result, we obtain bounds for the linearly parametrized model
  (\ref{eq:BasicModelBasis}), and the linear model
  (\ref{eq:BasicModelLin}).\\
  \textbf{Upper bound via regularized least squares:} For the linear model
  (\ref{eq:BasicModelLin}), and suitable class of sparse matrices
  $\Param^0$, we prove high-dimensional consistency of the penalized
  least-squares method introduced in Section \ref{sec:ResContinuumRegularized}.
  The resulting upper bound on sample-complexity
  matches the information theoretic lower bound up to constant factors in $p$. 

For the sake of simplicity, in this section we focus on the case of sparse linear SDEs, stating upper and lower bounds, cf. Section \ref{sec:SparseLinear}.
We then illustrate the general theory by analyzing a specific but rich problem: learning the Laplacian of a sparse graph, cf. Section \ref{sec:LearningLaplacian}. In Section \ref{sec:num_resul_main_th} we give numerical illustrations of our main results. Extensions, in particular, general lower bounds on the sample
complexity, are discussed in Section \ref{sec:extensions}.
Finally, in Section \ref{sec:num_ext}, we present numerical illustrations of these extensions, part of which are motivated by real-world applications. 

Proofs for the technical lemmas are provided in the appendix.

%
%*********************************************************
%
\subsection{Notation}

For any $N \in\naturals$, we let $[N] = \{1,2,\dots,N\}$.

Given any matrix $Q$, its transpose is denoted by $Q^*$
and its support, $\supp(Q)$, is the $0-1$
matrix such that $\supp(Q)_{ij}=1$ if and only if $Q_{ij}\neq 0$. 

For a vector $v\in\reals^N$, $\supp(v)$ is defined analogously.
With a slight abuse of notation, we occasionally
write $\supp(v)$ for the subset of indices $i\in [N]$ such that
$v_i\neq 0$. The \emph{signed support} of a matrix (or vector) $Q$,
denoted by $\sign(Q)$, is the matrix  defined by 
$\sign(Q)_{ij}=\sign(Q_{ij})$ where the function $\sign(Q_{ij})$ is defined as
\begin{equation}
  \sign(Q)_{ij} = \left\{
    \begin{array}{l l}
      +1 & \qquad \text{if } Q_{ij} > 0 \\
      0 & \qquad \text{if } Q_{ij} = 0 \\ 
      -1 & \qquad \text{if } Q_{ij} < 0 
    \end{array}
  \right.
\end{equation}
The $r$-th row of a matrix $Q$ is denoted
by $Q_r$.
Given a matrix $Q\in\reals^{M\times N}$, and sets $L\subseteq [M]$,
$R\subseteq [N]$, we denote by $Q_{L,R}$ the sub-matrix $Q_{L,R}
\equiv (Q_{ij})_{i\in L, j\in R}$.

For $q\ge 1$, the $\ell_q$ norm of a vector $v\in\reals^N$ is given by
$\|v\|_q \equiv(\sum_{i\in [N]} |v_i|^q)^{1/q}$. This is extended in
the usual way to $q=\infty$. 
As usual, the misnomer `$0$-norm' is used for the size of the
support of $v$, namely $\|v\|_0$ is the number of non-zero entries of $v$. 
The $\ell_q$ operator norm of a matrix
$Q\in\reals^{M\times N}$ is denoted by $\III Q \III_{q}$.
In particular the
$\ell_{\infty}$ operator norm is given by
$\III Q \III_{\infty}\equiv \max_{r\in [M]} \|Q_r\|_1$.

If $Q\in\reals^{N\times N}$ is symmetric, then its eigenvalues are
denoted by $\lambda_1(Q)\le \lambda_2(Q)\le \dots \le \lambda_N(Q)$.
The minimum and maximum eigenvalues are denoted as
$\lambda_{\rm min}(Q) \equiv\lambda_1(Q)$ and $\lambda_{\rm max}(Q)
\equiv \lambda_N(Q)$.
For a general (non-symmetric)  matrix $Q\in\reals^{M\times N}$ we let
$0\le \sigma_1(Q)\le \dots\le \sigma_{\min\{M, N\}}(Q)$ denote its
singular values. Further $\sigma_{\rm min}(Q) =\sigma_1(Q)$ and
$\sigma_{\rm max}(Q)=\sigma_{\min\{M, N\}}(Q)$ are the minimum and
maximum singular values. 

Throughout the paper, we denote by $C$, $C_1$, $C_2$, etc, constants
that can be adjusted from point to point.
%
%*********************************************************
%
\subsection{Sample complexity for sparse linear SDEs}
\label{sec:SparseLinear}

In order to state our results, it is convenient to define the class of
sparse matrices $\cA^{(S)}$,  depending on parameters $k,p\in\naturals$,
$k\ge 3$, $\param_{\rm min},\rho_{\min}>0$,
\begin{equation}
\cA^{(S)}= \cA^{(S)}(k,p,\param_{\rm min},\rho_{\min}) \subseteq \reals^{p\times p}
\end{equation}
by letting  $\Param \in \cA^{(S)}$ if and only if
\begin{itemize}
\item[$(i)$]  $\|\Param_r\|_0 \leq k$  for all  $r\in [p]$.
\item[$(ii)$] $|\param_{ij}| \ge \param_{\min}$ for all  $i,j\in [p]$ such
  that $\param_{ij} \neq 0$.
\item[$(iii)$] $\lambda_{\min}(-(\Param+\Param^*)/2)  \ge
  \rho_{\min} >0$. 
\end{itemize}
Notice in particular that condition $(iii)$ implies that the system of
linear ODEs $\dot{x}(t) = \Param x(t)$  is stable. Equivalently, the
spectrum of $\Param$ is contained in the half plane $\{z\in\complex\, :\, {\rm Re}(z)<0\}$. 
As a consequence, if $\Param^0\in \cA^{(S)}$, then the diffusion process (\ref{eq:BasicModelLin})
has a unique stationary measure which is Gaussian with covariance $Q^0\in\reals^{p\times p}$
and is given by the unique solution of Lyapunov's equation  \cite{Control}
\begin{equation}
\Param^0 Q^0 +  Q^0 (\Param^0)^* + I = 0 \label{eq:Lyapunov}.
\end{equation}
Hence $X^T_0 = \{x(t):\; t\in [0,T]\}$ is a stationary trajectory distributed
according to the linear model (\ref{eq:BasicModelLin}) if 
$x(t=0)\sim\normal(0,Q^0)$ is a Gaussian random variable independent of $b(t)$.

We consider the linear model \eqref{eq:BasicModelLin} with $\Param^0
\in \cA^{(S)}$. Given a row index $r\in [p]$, let $S^0 = S^0(r)$
be the support of $\Param^0_r$.\\
\textbf{Assumption 1 (Restricted convexity).} For $C_{\rm min}>0$, we have
\begin{eqnarray}
 \lambda_{\min} (Q^0_{S^0,S^0})\ge C_{\rm min}\, .
\end{eqnarray}
\textbf{Assumption 2 (Irrepresentability):} For some $\alpha>0$, we have
\begin{equation}
\III  {Q^0}_{(S^0)^C, {S^0}} \left({Q^0}_{{S^0},{S^0}}\right)^{-1}  
\III_\infty \le  1 - \alpha\, .
\end{equation}
We refer to \cite{zhao,Buhlmann} for the original development of these
conditions in the context of sparse regression.

Our first theorem establishes high-dimensional consistency of
$\ell_1$-penalized least squares for estimating
$\sign(\Param^0)$
from a stationary trajectory $X^T_0$ according to
the linear model (\ref{eq:BasicModelLin}) when
$\Param^0 \in \cA^{(S)}$.

\begin{theorem}\label{th:main_cont_time}
If $\Param^0 \in \cA^{(S)}(k,p,\param_{\min},\rho_{\min})$
satisfies assumptions $1$ and $2$ above for all $r \in [p]$
and some $C_{\min},\alpha > 0$, then there
exists $\lambda = \lambda(T) > 0$ such that
\begin{equation} 
T_{\rls(\lambda)}(\Param^0) \leq
\frac{2 \cdot 10^4 k^2 (k \, \rho_{\min}^{-2}
+ \param_{\min}^{-2})}{\alpha^{2}
\rho_{\min} C_{\min}^{2} } \, \log\Big(\frac{4p
k}{\delta}\Big).
 \label{eq:sample_bound_cont}
\end{equation}
In particular, one can choose
\begin{equation} 
\lambda = 
 \sqrt{\frac{36}{T\alpha^2 \rho_{\min}}\,
\log\Big(\frac{4p}{\delta}\Big)}\,\, .
\end{equation} 
\end{theorem}
\begin{remark}
Note that the notions of sample-complexity
introduced in \ref{eq:def_sample_complexity} and 
\ref{eq:def_sample_complexity_group} are well-defined
for reconstruction algorithms that depend on $T$,
the length of the stationary trajectory $X^T_0$.
This is the case with the regularized least squares algorithm
$\rls(\lambda)$, since $\lambda$ can depend on $T$.
\end{remark}

\begin{remark}
If there exists $C_{\min},\alpha > 0$ such that
assumptions 1 and 2 hold for all $r \in [p]$
and for all $\Param^0 \in \cA^{(S)}(k,p,\param_{\min},\rho_{\min})$,
then we can replace $T_{\rls(\lambda)}(\Param^0)$
by $T_{\rls(\lambda)}(\cA^{(S)})$ in \eqref{eq:sample_bound_cont}.
\end{remark}

The next theorem establishes a lower bound on the
sample-complexity
of learning the signed support of $\Param^0 \in \cA^{(S)}$ from
a stationary trajectory, $X^T_0$, distributed
according to the linear model (\ref{eq:BasicModelLin}).
\begin{theorem} \label{th:linear_lbound_sparse}
Let $\alg = \alg(X^T_0)$ be an estimator of $\sign(\Param^0)$.
There is a constant $C(k,\delta)$, such that,
for all $p$ large enough,
\begin{align}\label{eq:lower_bound_linear_model}
T_{\alg}(\cA^{(S)})
\ge C(k,\delta)\, \max \Big\{ \frac{\rho_{\min} }
{\param_{\min}^2},\frac{1}{\param_{\min}}\Big\} \log p\, .
\end{align}
\end{theorem}

\begin{remark}
Theorem \ref{th:linear_lbound_sparse}
cannot be used to conclude that,
if $T$ is `small', then $\rls(\lambda)$ always fails
to reconstruct $\Param^0$ from $X^T_0$
regardless of the choice of $\lambda$. What the lower bound
says is that, if $T$ is `small', then, for every choice of
$\lambda = \lambda(X^T_0)$, there exists a
$\Param^0 \in \cA^{(S)}$ that cannot be reconstructed.
The particular $\Param^0$ that cannot be reconstructed, however,
can depend on the choice of $\lambda$.
\end{remark}

These two theorems establish that, under assumptions 1 and 2 above,
the time-complexity of learning the signed support of the diffusion
coefficient for sparse linear SDEs in the class $\cA^{(S)}$
is $\bigo(\log p)$.

Notice that both upper and lower bounds depend
in a non-trivial way on the parameter $\rho_{\min}$.
In order to gain intuition on this quantity,
consider Eq.~(\ref{eq:BasicModelLin}) in absence of the driving term $\de b_i(t)$.
By using the Lyapunov function
$\|x(t)\|_2^2$, it is  easy to verify that $\|x(t)\|_2\le \|x(0)\|_2\,
e^{-\rho_{\rm min}t/2}$. Hence $\rho_{\rm min}^{-1}$ provides a general upper
bound on the mixing time of the diffusion (\ref{eq:BasicModelLin}).
The upper bound is essentially tight if the matrix $\Param^0$ is
symmetric.

Theorems \ref{th:main_cont_time} and
\ref{th:linear_lbound_sparse} can therefore be used to characterize the
dependence of the sample complexity on the mixing time. 
One subtle aspect is that $C_{\rm min}$ and $\rho_{\rm min}$ cannot be
varied independently because of  the Lyapunov equation,
Eq.~(\ref{eq:Lyapunov}).
In order to clarify this dependency, we apply our
general results to the problem of learning the
Laplacian of an undirected graph. 
%
%*********************************************************
%

\subsection{Learning the laplacian of graphs with bounded degree}
\label{sec:LearningLaplacian} 

Given a simple graph $\cG=(\cV,\cE)$ on 
vertex set $\cV=[p]$, its Laplacian $\Delta^{\cG}$  is
the symmetric $p\times p$ matrix which is equal to the adjacency matrix
of $\cG$ outside the diagonal, and with entries $\Delta^{\cG}_{ii}=
-\dg(i)$ on the diagonal \cite{Chung}. 
(Here $\dg(i)$ denotes the degree of vertex $i$.)

It is well known that $\Delta^{\cG}$ is negative semidefinite, 
with one eigenvalue equal to $0$, whose multiplicity is equal to the
number of connected components of $\cG$.
The matrix $\Param^0 = -m\, I + \Delta^{\cG}$ fits into the setting of
Theorem \ref{th:main_cont_time} for $m>0$. 
The corresponding model \eqref{eq:BasicModelLin} describes the 
over-damped dynamics of a network of masses connected by springs
of unit strength, and connected by a spring of strength $m$ to the origin.

Let $\cG_{\text{bounded}} =
\cG_{\text{bounded}}(k,p)$ be the class of graphs on $p$ nodes
with maximum vertex degree bounded by $k$. Define,
\begin{align}
& \cA^{(L)}(m, p, k) = \nonumber \\
& \quad \{\Param^0 = -m\, I + \Delta^{G} \;|\; m > 0,  G \in \cG_{\text{bounded}}\}	
\end{align}
The following theorem holds regarding the sample-complexity
of learning the signed support
of $\Param^0$ from a stationary trajectory $X^T_0$ of a linear
SDE with $\Param^0 \in \cA^{(L)}$.

\begin{theorem}\label{th:cont_reg_graph_bounded_degree}
If $\Param^0 \in \cA^{(L)}(m, p, k)$
then there exists $\lambda = \lambda(T) > 0$ such that
\begin{equation*}
T_{\rls(\lambda)}(\cA^{(L)}) \leq 4\cdot  10^5 k^2\,
\Big(\frac{k+m}{m}\Big)^{5} (k + m^2)\, 
\log \Big(\frac{4p k}{\delta}\Big)\, ,
\end{equation*}
In particular one can take,
$$\lambda = \sqrt{36(k+m)^2\log(4p/\delta)/(T m^3)}.$$
\end{theorem}

In other words, for $m$ bounded away
from $0$ and $\infty$, regularized least squares regression correctly
reconstructs the graph $\cG$ from a trajectory of time length 
which is polynomial in the degree and logarithmic in the graph size.

Using this theorem we can write the following corollary that helps compare the bounds
obtained in Theorems \ref{th:main_cont_time} and \ref{th:linear_lbound_sparse} above.

\begin{corollary}\label{th:laplace_large_m}
Assume the same setting as in Theorem
\ref{th:cont_reg_graph_bounded_degree}.
There exist constants $\lambda = \lambda(T)$,
$C_1=C_1(k,\delta)$ and $C_2=C_2(k,\delta)$
such that, for all $p$ large enough,
\begin{align*}
m< k & \;\;\Rightarrow \;\;\;\;\;C_1  \log p \le T_{\rls(\lambda)}(\cA^{(L)}) \le  C_2 m^{-5} \log p, \\
m\ge k  &\;\;\Rightarrow\;\; C_1 m \log p \le T_{\rls(\lambda)}(\cA^{(L)}) \le  C_2 m^2 \log
p\, .
\end{align*}
In addition, the lower-bounds hold regardless of the choice of $\lambda$.
\end{corollary}
\begin{proof}
The proof of this corollary follows immediately from
Theorem \ref{th:cont_reg_graph_bounded_degree} and Theorem \ref{th:linear_lbound_sparse}.
\end{proof}

Notice that the upper bound on $T_{\rls}$ presents a non-trivial behavior
in $m$. It diverges both at large $m$, and at small $m$. The reasons
of these behaviors are different. For small $m$, the mixing time of the
diffusion (which is proportional to $1/m$) gets large, and hence a large
time is necessary to accumulate information about
$\Param^0$. Vice-versa for large $m$, $\Theta^0$ gets close to $-m\, I$
and hence it depends weakly on the graph structure. 

Notice that the lower bound also diverges as $m\to\infty$, hence
confirming the above picture. On the other hand, the behavior of $T_{\rls}$
as $m\to 0$ remains an open question since our lower bound stays
bounded in that limit.

%
%*********************************************************
%

\section{Related work}
\label{sec:Related}

The problem of estimating the parameters
of a diffusion plays a central role in several applied domains,
examples being econometrics, chemistry and system biology.

In the first context, diffusions are used to model the evolution 
of price indices \cite{PhilipsYu}. While the most elementary process
is the (geometric) Brownian motion \cite{Bachelier,BlackScholes},
a number of parametric families have been introduced to account for
nonlinearities. The number of parameters is usually small and
parameter estimation is addressed via maximum
likelihood (ML). We refer to \cite{brown1975asymptotic,Basawa,Kutoyants} for proofs of consistency and
asymptotic normality of the ML estimator.
Much of the recent research has focused on dealing with the challenges
posed by the fact that the diffusion is sampled at discrete intervals,
and the transition probabilities cannot be computed in closed form. A
short list of contributions on this problem includes
\cite{lo1986maximum,Dacunha,PedersenFirst,Sahalia}. 
In particular, asymptotically consistent methods based on approximate
transition probabilities exist, see for instance \cite{PedersenSecond,ChangChen}.
Nonparametric estimation of the
drift coefficient has been studied as well
\cite{Fan,Spokoiny,Dalalyan}.

However, all of these works focus on the low-dimensional
setting: the vector of parameters to be estimated is $p$-dimensional,
and the diffusion is observed for a time $T\to\infty$. 
Hence there is little overlap with the present work. 
In particular, simple ML estimators are not viable in the
high-dimensional setting.
At the same time, it would be interesting to address the problems posed by discrete sampling and
non-parametric estimation in the high-dimensional setting as well.

Applications to chemistry and system biology have been mentioned 
in Section \ref{sec:Intro}. A large variety of chemical reactions are 
modeled by diffusions with suitably parametrized drift terms \cite{Gillespie,Higham}. 
Of particular interest here are special classes of drift coefficients,
for instance those exhibiting time-scale separation \cite{Stuart} or
gradients of a potential \cite{RemarksDrift}. \cite{oates2012network} use
regularized least squares to learn SDEs and from them recover both intracellular
and intercellular biological networks. In this work, several regularizations are
studied, including $\ell$-1 regularization, but no guarantees are proved. In a
different work, \cite{valdes2005estimating}, the same method is applied to study
the functional connectivity of the brain. As with the econometrics applications, these works have focused on low-dimensional diffusions.

Technically, our work fits on recent developments in learning
high-dimensional graphical models.
The typical setting assumes that the data are \emph{independent and identically distributed} (i.i.d.) samples from a
high-dimensional Gaussian distribution with sparse inverse covariance.
The underlying graph structure (the support of the inverse covariance)
is estimated using convex regularizations that promote sparsity. 
Well known examples include the \emph{graphical} LASSO \cite{friedman2008sparse}
and the pseudo-likelihood method of \cite{Buhlmann}.
In the context of binary pairwise graphical models, similar methods
were developed in \cite{wainwright2007high}.

More closely related to our paper is the work reported in \cite{nowak2012}. 
It proposes an algorithm to learn the interference graph in a wireless network
from passive measurements of the traffic.
The paper is concerned with the number of samples required in order
to recover the interference graph correctly.
Both information theoretic lower bounds and upper bounds using a practical
algorithm are provided.
The model used in this work is a time-evolving discrete time model and the algorithm
 is domain specialized. 
In contrast, the emphasis of our work is on the continuous time models and
indeed a significant portion of our effort is dedicated to obtaining the right
scaling in this scenario.
Furthermore, the algorithms analyzed  in these two works are completely
different.

To the best of our knowledge the present work is the first one moving beyond the
assumption of independent samples from a continuous time diffusion process when
dealing with the sample complexity of learning the structure of the underlying graph. While we extend ideas and methods from this literature, dealing with dependent samples raises new mathematical challenges.

Our methods  build on the work on $\ell_1$-regularized least squares, and its variants
\cite{tibshirani_lasso, donoho2006nearsol, donoho2006sol, zhang2009ssp, wainwright2009sth}.
The most closely related results are the one concerning high-dimensional 
consistency for support recovery \cite{Buhlmann,wainwright2007high,zhao}. 
Our proof for our upper bound follows indeed the approach developed in these papers,
with two important challenges. First, the design matrix in our case is produced by a stochastic diffusion, and it does not necessarily
satisfy the irrepresentability conditions used by these works. Second,
the observations are not independent  and therefore elementary concentration 
inequalities are not sufficient.

Most of these proofs build on the technique of \cite{zhao}.
A naive adaptation to the present case 
allows to prove some performance guarantee  
for the discrete-time setting. However the resulting 
bounds are not uniform as the sampling interval  $\eta$ tends to $0$ for $n\eta = T$ fixed.
In particular, they do not allow to prove an analogous of
our continuous time result, Theorem \ref{th:main_cont_time}.
A large part of our effort is devoted to proving more accurate 
probability estimates that capture the correct scaling for small
$\eta$.

Finally, the related topic of learning graphical models for
autoregressive processes was studied recently in 
\cite{haufe2009sparse,Songsiri1,Songsiri2}. These papers propose a convex relaxation
that is different from the one studied in this paper, without however
establishing high-dimensional consistency for  
model selection.

Preliminary report of our work were presented at NIPS 2010 \cite{bento2010learning}
and ISIT 2011 \cite{Bento_IT}. Subsequent work by Bolstad, Van Veen
and Nowak \cite{Bolstad} establishes high-dimensional consistency for
estimating autoregressive models through a related approach. These
guarantees are non-uniform in the sampling rate $\eta$.
The work of \cite{loh2011high} provides upper bounds on the error of regularized least square when observations are not independent. Although bounding the error of $\rls$ is related to our problem of support recovery, in the context of learning SDEs, the conditions under which their result holds are never reduced or related to properties of the dynamics of the SDE alone. In addition, it is unclear whether their conditions hold uniformly with the sampling rate $\eta$( the results presented only apply directly to discrete time). The more recent work of \cite{ibrahimi2012efficient} relates to ours by showing that, under suitable conditions, sparse linear quadratic systems can be estimated and adaptively controlled with few observations. Finally, \cite{no2013minimax} provides a framework for filtering $X^T_0$ which could be used to estimate $\Param^0$. It is an interesting open problem to investigate how an estimator obtained from their framework compares to ours.

%
%****************************************************************
%

\section{Numerical illustrations of the main theoretical results} \label{sec:num_resul_main_th}

In this section we illustrate our main results on
synthetic data.  These numerical results agree
with our observations in Theorems \ref{th:main_cont_time}, \ref{th:linear_lbound_sparse} and 
\ref{th:cont_reg_graph_bounded_degree} 
that the time-complexity for learning linear sparse SDEs
scales logarithmically with the number 
of nodes in the network $p$, given a constant maximum degree.
They also agree with the implication of Theorem \ref{th:main_discrete} that
the time-complexity is roughly independent of the sampling rate,
assuming that we are in the regime of small $\eta$. Or, in other words,
that our reconstruction guarantees are uniform in the sampling rate
for small $\eta$.
%
%\vspace{-2.5em}
%
\begin{figure}[H]
\centering
\includegraphics[width = .39 \textwidth]{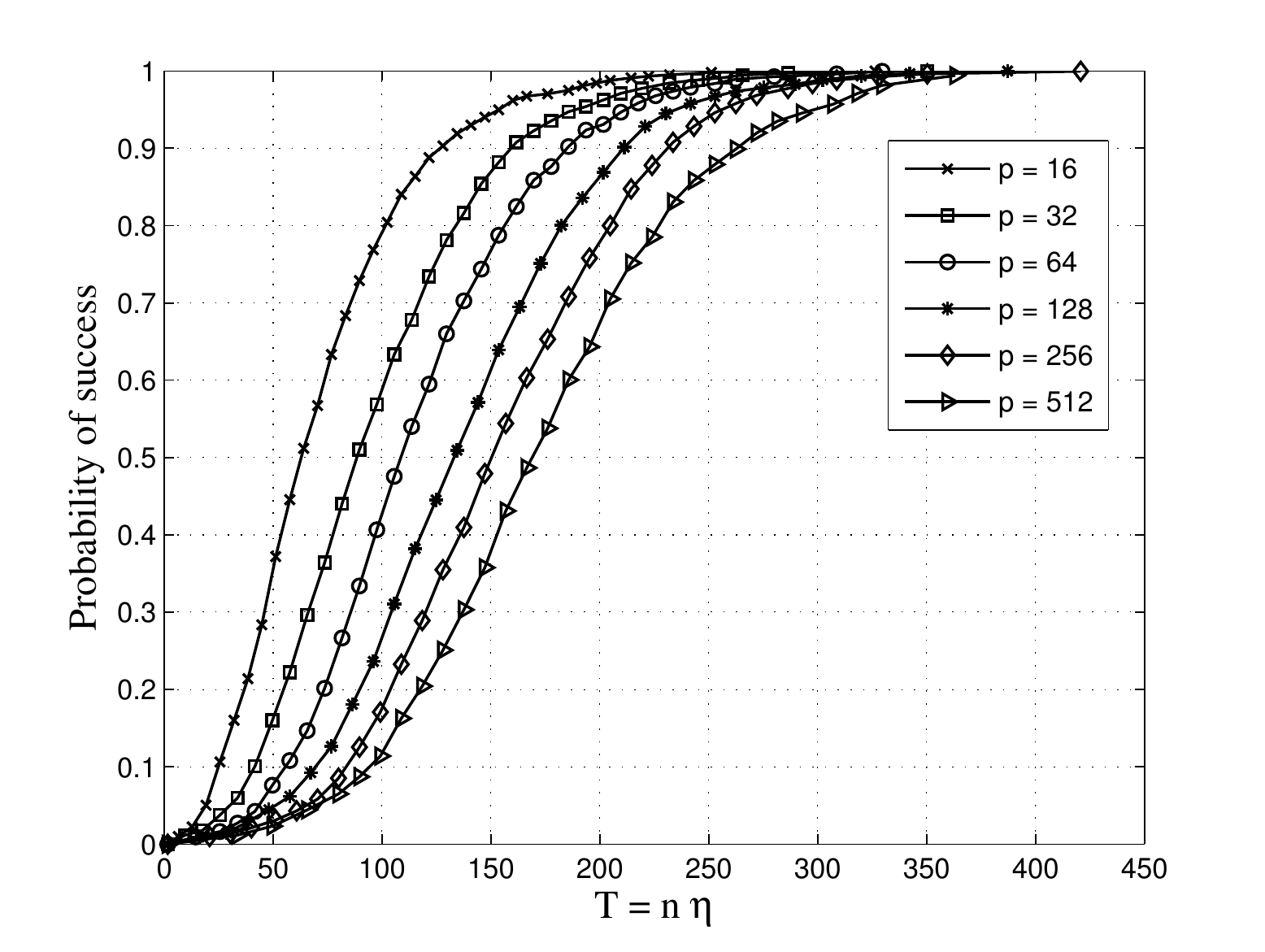}\\
\includegraphics[width =  .39 \textwidth]{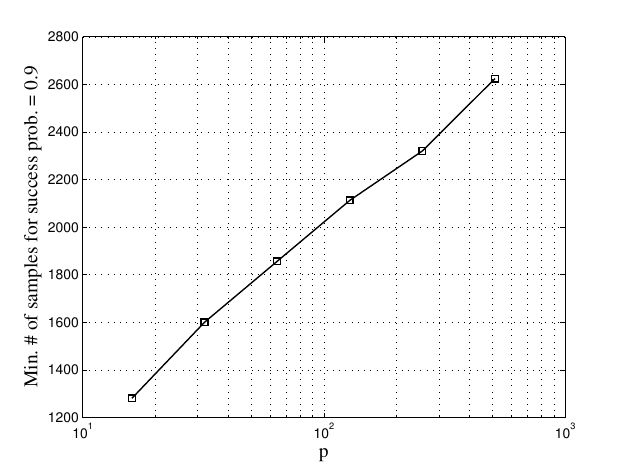}
\caption{(top) Probability of success vs. length of the observation interval $n\eta$. (bottom) Sample complexity for 90\% probability of success vs. p.}
\label{fig:psVsN}
\end{figure}
%
%\vspace{-2.0em}
%

Note that, in order to obtain numerical values for the time-complexity
that do not depend on $\lambda$, we use a definition for
sample-complexity and time-complexity that is slightly different
than the one used when stating our main results. 

We start by analyzing the performance of $\rls$ for the discrete analogue of \eqref{eq:BasicModelLin} (See
equation \eqref{eq:DiscreteTimeModel} in Section \ref{sec:extensions}).
Our results are summarized in Figures \ref{fig:psVsN} and \ref{fig:psVsNeta_eta2}. First, we generate data as follows. We draw  $\tilde{\Param}^0$ as a random sparse matrix in $\{0,1\}^{p \times p}$ with elements chosen 
independently at random with $\prob(\param^0_{ij}=1) = k/p$, $k=5$, and form
$\Param^0 = -7 \mathbb{I} + \tilde{\Param}^0$
\footnote{For $p$ large, the SDE generated is stable with high-probability.}.
Second, a sample path $X^n_0\equiv\{x(t): 0\le t\le n \}$ is obtained from Eq.~\eqref{eq:DiscreteTimeModel}.
Finally, we choose an $r \in [p]$ uniformly at random and 
solve the regularized least squares problem 
\footnote{For discrete-time SDEs, the cost function
is given explicitly in Eq.~\eqref{eq:discr_reg_prob}.}
for a different number of observations $n$
and different values of $\lambda$. We record a $1$ or a $0$ if
the correct signed support of $\Param^0_r$ is recovered or not.
For every value of $n$ and $\lambda$,
the probability of successful
recovery is then estimated by taking the average
of these errors over all realizations of $\Param^0$, $X^n_0$ and $r$.
Finally, for each fixed $n$, we take the maximum over $\lambda$ of
these probability of success.
The top plot in Figure \ref{fig:psVsN} depicts the probability of success vs. $n\eta$ 
for $\eta =0.1$ and different values of $p$. Each curve is obtained using
$2^{11}$ instances, and each instance is generated using a new random 
matrix $\Param^0$.
In addition, from this plot of $n$ vs. probability of success, we generate
the bottom plot in Figure \ref{fig:psVsN}: sample-complexity vs. $p$. 
To be explicit, the definition of sample-complexity in use is
%
%\vspace{-2.5em}
%
\begin{align} \label{eq:sample_complex_def_for_num_res}
N_{\rls}(\cA) = \inf \{ & n_0 \in \naturals_0 : 
\sup_{\lambda>0} \hat{\E} \{ \hat{\prob}_{\Param^0,n} \{\rls(\lambda) = \sign(\Param^0) \} \} \nonumber \\
& \qquad \geq 1- \delta \text{ for all } n \geq n_0\},
\end{align}
where we choose a probability of success of $\delta =0.9$.
Above, $\hat{\E}$ represents empirical expectation over $\Param^0$
and $\hat{\prob}$ empirical probability over $X^n_0$, and,
$\cA$ is the class of all matrices that can be generated
by the random procedure described before.
In agreement with Theorem \ref{th:cont_reg_graph_bounded_degree},
the curve shows the logarithmic scaling of the sample-complexity with $p$. 

In Figure \ref{fig:psVsNeta_eta2} we turn to the continuous-time model 
(\ref{eq:BasicModelLin}). Trajectories are generated by `discretizing'
this stochastic differential equation with step $\eta'$ much smaller
than the sampling rate $\eta$. 
We draw random matrices $\Param^0$ as above and plot the probability of success 
for $p=16$, $k=4$ and different values of $\eta$,
as a function of $T$. We used $2^{11}$ instances for each curve.
The time-complexity in use for these plots is
the continuous-time analog of \eqref{eq:sample_complex_def_for_num_res}.
Again in agreement with Theorem \ref{th:main_discrete},
for a fixed observation interval $T$, the probability of success 
converges to some limiting value as $\eta\to 0$.

\begin{figure}[H]
\centering
\includegraphics[width = .39 \textwidth]{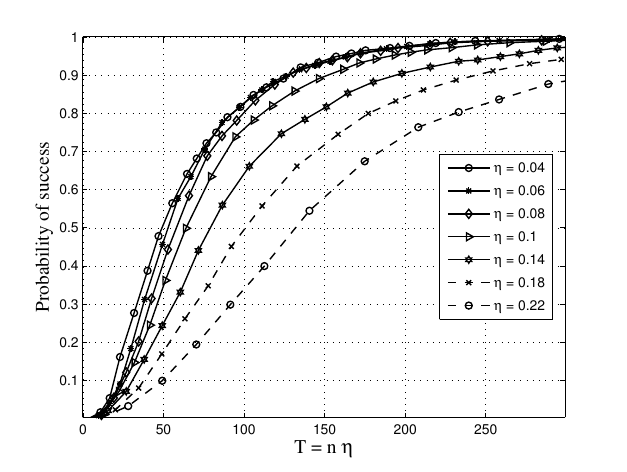}\\
\includegraphics[width = .39 \textwidth]{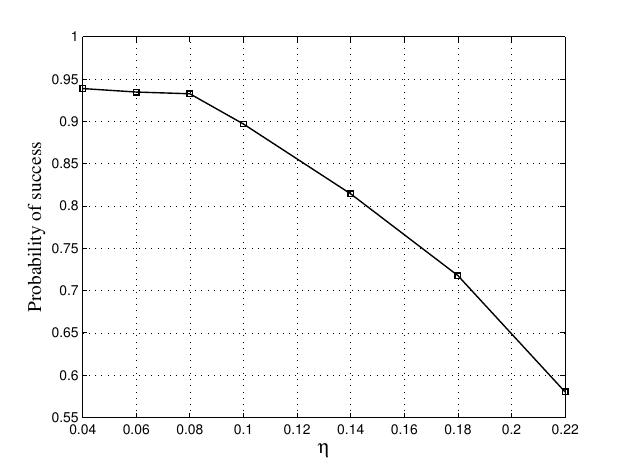}
\caption{(top) Probability of success vs. length of the observation interval $n\eta$ for different values of $\eta$. (bottom) Probability of success vs. $\eta$ for a fixed length of the observation interval, ($n\eta = 150$) . The process is generated for a small value of $\eta$ and sampled at different rates.}
\label{fig:psVsNeta_eta2}
\end{figure}

%
%****************************************************************
%

\section{Extensions}
\label{sec:extensions}

In this section we present some extensions to our previous results.
We begin by presenting an analogous theorem of Theorem \ref{th:main_cont_time}
for the case of a discrete time system. This is an important result in
itself and also constitutes the basis for
the proof of Theorem \ref{th:main_cont_time}.
In fact, Theorem \ref{th:main_cont_time} is
proved by letting $\eta \rightarrow 0$ in the result below.
We then present a general lower bound on the time-complexity
of learning continuous stochastic differential equations. Using
this result, lower bounds for the time-complexity of
linear SDEs with dense matrices $\Param^0$ and non-linear SDEs are derived.

\subsection{Discrete-time model}
\label{sec:DiscreteResults}
The problem of learning stochastic differential equations in discrete time
is important in itself and also because it relates to the problem of learning
a continuous-time stochastic differential equation from discretely sampling its
continuous trajectory. Focusing on continuous-time dynamics allowed us to
obtain the elegant statements of Section \ref{sec:SparseLinear}. However, much
of the theoretical analysis concerning the regularized least square algorithm
is in fact devoted to the analysis of the following discrete-time dynamics, with parameter $\eta>0$:
\begin{eqnarray}
x(t) = x(t-1)+\eta \Param^0 x(t-1) + \, w(t), \;\;\;\;\; t \in \naturals_0\, .
\label{eq:DiscreteTimeModel}
\end{eqnarray}
Here $x(t)\in \reals^p$ is the vector collecting the dynamical
variables, $\Param^0  \in\reals^{p\times p}$ specifies the dynamics as 
above, and $\{w(t)\}_{t\ge 0}$ is a sequence of i.i.d.
normal vectors with covariance $\eta\, I_{p\times p}$ (i.e.
with independent components of variance $\eta$).
We assume that $n+1$ consecutive samples are
given, $X^n_0 \equiv \{x(t):  0 \leq t \leq t \}$,
and ask under which conditions 
regularized least squares reconstructs the signed support of $\Param^0$. 

The parameter $\eta$ has the meaning of a time-step size. The 
continuous-time model (\ref{eq:BasicModelLin}) is recovered,
in a sense made precise below, by letting $\eta\to 0$. Indeed,
for this discrete time model, we prove reconstruction guarantees
that are uniform in this
limit as long as the product $n \eta$ (which corresponds to
the time interval $T$ in the Section \ref{sec:SparseLinear} ) is kept constant.
For a formal statement we refer to Theorem \ref{th:main_discrete}.
Theorem \ref{th:main_cont_time} is indeed proved by carefully 
controlling this limit. The mathematical challenge in this problem
is related to the fundamental fact that the samples $\{x(t)\}_{0\le t\le n}$
are dependent (and strongly dependent as $\eta\to 0$).

Discrete time models of the form (\ref{eq:DiscreteTimeModel}) can arise 
either because the system under study evolves by discrete steps,
or because we are sub-sampling a continuous time system
modeled as in Eq.~(\ref{eq:BasicModel}). Notice that in the latter case 
the matrices $\Param^0$  appearing in Eq.~(\ref{eq:DiscreteTimeModel})
and (\ref{eq:BasicModel}) coincide only to the zeroth order in $\eta$.
Neglecting this technical complication, the uniformity of 
our reconstruction guarantees as $\eta\to 0$ has an appealing
interpretation already mentioned above. Whenever the samples spacing
is not too large, the time-complexity (i.e. the product $n\eta$) 
is roughly independent of the spacing itself.

Consider a system evolving in discrete time according to 
the model (\ref{eq:DiscreteTimeModel}),
and let $X^n_0$ be the observed portion of the trajectory.
The $r^{\text{th}}$ row of $\Param^0$, $\Param^0_r$,
is estimated by solving the following 
convex optimization problem
\begin{equation}
\underset{\Param_r\in\reals^p}{\rm minimize}\;\;\;
 \cL(\Param_r;X^n_0) + \lambda \| \Param_r\|_1\, ,\label{eq:discr_reg_prob}
\end{equation}
where the log-likelihood function $\cL(\Param_r;X^n_0)$ is defined as 
\begin{equation}
\frac{1}{2\eta^2 n}\, \sum_{t=0}^{n-1}
\left\{x_r(t+1)-x_r(t)-\eta\, \< \Param_r, x(t)\>\right\}^2\, .
\label{eq:opt_prob}
\end{equation}
Apart from an additive constant, the $\eta\to 0$ limit of this cost 
function can be shown to coincide with the cost function 
in the continuous time case, cf. Eq.~(\ref{eq:ContCost}). Indeed the proof
of Theorem \ref{th:main_cont_time} will amount to a more precise version
of this statement.
Furthermore, $\cL(\Param_r;X^n_0)$ is easily seen to be the log-likelihood of
$\Param_r$ within model (\ref{eq:DiscreteTimeModel}).

Let us introduce the class of sparse matrices $\cA'^{(S)}$ as being exactly equal to
the class  $\cA^{(S)}$ introduced in Section \ref{sec:SparseLinear} but with
condition $(iii)$ replaced by
\begin{equation}
\frac{ 1-\sigma_{\max} (I+\eta\, \Param^0) }{\eta} \geq D > 0
\end{equation}

If $\Param^0 \in \cA'^{(S)}$ then, under the model (\ref{eq:DiscreteTimeModel}), $x(t)$ has
a unique stationary measure which is Gaussian with covariance $Q^0$ determined
by the following modified Lyapunov equation
\begin{equation}
\Param^0 Q^0 +  Q^0 (\Param^0)^* + \eta \Param^0 Q^0 (\Param^0)^* +I = 0 \, .
\label{eq:Lyapunov2}
\end{equation}
It will be clear from the context whether $\Param^0$ (or $Q^0$) refers to the dynamics matrix (or covariance of the stationary distribution) from the continuous or discrete time system.

The following theorem establishes the conditions under which $\ell_1$-regularized least squares recovers $\sign (\Theta^0)$ with high probability. 
\begin{theorem}\label{th:main_discrete}
Assume that $\Param^0 \in \cA'^{(S)}(k,p,\param_{\min},D)$
and that $\Param^0_r$ satisfies assumptions 1 and 2 of Section \ref{sec:SparseLinear}.
Let $X^n_0$ be a stationary trajectory
distributed according to the linear model \eqref{eq:DiscreteTimeModel}. If
\begin{equation} 
n \eta > \frac{10^4  k^2 (k D^{-2} + \param_{\min}^{-2})}{\alpha^{2} D C_{\min}^{2}}
\, \log\Big(\frac{4p k}{\delta}\Big)\, ,\label{eq:sample_bound_disc}
\end{equation}
then there exists $\lambda = \lambda(n \eta) > 0$ such that 
$\ell_1$-regularized least squares recovers the 
signed support of $\Param^0_r$ with probability larger than $1-\delta$.
This is achieved by taking
$\lambda = \sqrt{(36\,\, \log(4p/\delta))/(D\alpha^2 n\eta)}$.
\end{theorem}
In other words the discrete-time sample complexity, $n$, is logarithmic 
in the model dimension, polynomial in the maximum network degree
and inversely proportional to the time spacing between samples.
The last point is particularly important. It enables us
to derive the bound on the continuous-time
sample complexity as the limit $\eta \rightarrow 0$ of
the discrete-time sample complexity. It also confirms our intuition mentioned
in the Introduction: although one can produce an arbitrary large number of
samples by sampling the continuous process with finer resolutions, there
is limited amount of information that can be harnessed from a given time interval $[0,T]$.

\begin{remark}
The form of Theorem \ref{th:main_discrete} is different than that of
Theorem \ref{th:main_cont_time}. In Theorem \ref{th:main_discrete}
we do not compute a bound on
\begin{align*}
& N_{\rls(\lambda)}(\Param^0) \equiv \min \big\{ n_0 > 0: \prob_{\Param^0,n} \{ \sign(\hParam) \\
&\quad  = \sign(\Param^0)\} \geq 1 - \delta \text{ for all } n \geq n_0 \big\},
\end{align*}
the sample-complexity
of reconstructing $\sign(\Param^0)$, but rather a bound
on the sample-complexity of reconstructing
the signed support of a \emph{particular}
row $r$, $\sign(\Param^0_r)$.
Obviously, if assumptions
1 and 2 hold for the same constants
$C_{\min},\alpha > 0$ across $r \in [p]$, then replacing
$\delta$ by $\delta/p$ in \eqref{eq:sample_bound_disc}
allows us to use union
bound and conclude that there exists
$\lambda$ for which 
\begin{equation*}
N_{\rls(\lambda)}(\Param^0) \, \eta \leq
\frac{2 \cdot 10^4  k^2 (k D^{-2}
+ \param_{\min}^{-2})}{\alpha^{2} D C_{\min}^{2}}
\, \log\Big(\frac{4p k}{\delta}\Big)\,.
\end{equation*}
(Notice the factor of $2$). The reason why we present
Theorem \ref{th:main_discrete}
in a different
form is to emphasize the fact that the proofs
for the upper bounds are based on the success
of $\rls$ for reconstructing a particular row $r$.
\end{remark}

%
%*********************************************************************
%

\subsection{General lower bound on time-complexity}
\label{sec:general_lower_bound}

In this section we derive a general lower bound
on the minimum time $T$ required to learn a property $M(\Param^0)$ associated
to $\Param^0$ from a trajectory $X^T_0$ distributed according to the general model
\eqref{eq:BasicModel}. For our problem, $M(\Param^0)$ is the signed-support of $\Param^0$. However, the bound holds in general. This result is used afterwards to derive
lower bounds for the time-complexity of learning linear SDEs
with dense matrices $\Param^0$
(Section \ref{sec:linear_dense_lower_bound}) and for the time-complexity of learning
non-linear SDEs (Section \ref{sec:non_linear_lower_bound}).

The general form of the results in this section, and in the remainder of 
Section \ref{sec:extensions}, is as follow: If $\hM_T(X^T)$,
an estimator of $M(\Param^0)$ based on $X^T$, achieves successful
recovery with probability greater than $1/2$ for every $\Theta^0$ in
a class $\cA$, then $T$ must be greater then a certain value that is dependent
on properties of $\cA$ (cf. Theorems \ref{th:linear_lbound_dense} and \ref{th:non_linear_lbound}).
These results however are a corollary of a more relaxed
result (Theorem \ref{th:main_lbound} and Corollary \ref{th:main_simpler_lbound})
where we only require that the expected rate of miss-estimation
is small when $\Theta^0$ is drawn at random from the ensemble $\cA$.
Clearly, if an estimator performs well over all $\Theta^0 \in \cA$ then
it must also perform well in expectation regardless of the distribution
assumed over $\cA$.

Without loss of generality, in the remainder of Section \ref{sec:general_lower_bound},
the parameter $\Param^0$ is a random variable
chosen with some unknown prior distribution $\prob_{\Param^0}$ (subscript
will be often omitted). Also, in the following theorems we assume that $M(\Param^0)$ can be described by an alphabet $\mathcal{M}$
of finite size $|\mathcal{M}| < \infty$. For example, if $\Param^0 \in \reals^{p \times p}$ and $M(.) = \supp(.)$ then $\mathcal{M}$ can be a set of $2^{p^2}$ symbols, one per possible support of $\Param^0$. If $M(.) = \sign(.)$ then
$|\mathcal{M}| = 3^{p^2}$ symbols suffice to describe all
possible signed-supports of $\Param^0$.

\begin{remark}[\bf Special notation]
In this section we make a small change in our
notation. Outside Section \ref{sec:general_lower_bound},
where $\Param^0$ is a
matrix of real numbers, $\prob_{\Param^0}$ represents
a probability
distribution over $X^T_0$ parametrized by $\Param^0$. In this
section however, subscripts indicate that probabilities
and expectations
are to be taken with respect to the random variable
in the subscript.
Hence, $\prob_{\Param^0}$ is a probability distribution
{\bf \emph{for}} the
random variable $\Param^0$

Unless specified otherwise, $\prob$ and $\E$ denote probability
and expectation with respect to the joint law of $\{x(t)\}_{t\ge 0}$
and $\Param^0$. As mentioned above $X^T_0 \equiv \{ x(t): t \in [0, T] \}$
denotes the trajectory up to time $T$.
Also, we define the variance of a vector-valued random
variable as the sum of the variances over all components.
In particular,
\begin{equation*}
\Var_{\Param^0|X^t_0 } (F(x(t);\Param^0)) = \sum^p_{i = 1}  \Var_{\Param^0|X^t_0 } (F_i(x(t);\Param^0)),
\end{equation*}
where $\Var_{\Param^0|X^t_0 }$ is the variance with
respect to $\Param^0$ conditioned on $X^t_0$.
\end{remark}

The following general lower bound,
is a consequence of an identity between mutual information and the integral
of conditional variance proved by Kadota, Zakai and Ziv \cite{kadotamutinfo} and
a similar result by Duncan \cite{duncan}.
\begin{theorem}\label{th:main_lbound}
Let $X^T_0$ be a trajectory of system \eqref{eq:BasicModel} with initial state $x(0)$ for
a specific realization of the random variables $x(0)$ and $\Param^0$.
Let $\hM_T(X^T_0)$ be an estimator of
$M(\Param^0)$ based on $X^T_0$. If  
$\prob_{x(0),\Param^0,X^T_0}(\hM_T(X^T_0) \neq M(\Param^0) ) <\frac{1}{2}$ then
\begin{equation} \label{eq:main_bound}
T \geq \frac{2 H(M(\Param^0)) - \log (\mathcal{|M|})- 2 I(\Param^0;x(0)) -2}{\frac{1}{T} \int^T_0
  \E_{X^t_0}  \{ \Var_{\Param^0|X^t_0 } (F(x(t);\Param^0)) \} \de t}\, ;
\end{equation}
where $|\mathcal{M}|$ is the size of the alphabet
of $M(\Param^0)$. 
\end{theorem}
\begin{proof}
Equation \eqref{eq:BasicModel} can be regarded as
describing a white Gaussian channel with feedback
where $\Param^0$ denotes the message to be
transmitted. For this scenario, Kadota et al. \cite{kadotamutinfo}
give the following identity for the mutual information between
$X^T_0$ and $\Param^0$ when the initial condition is $x(0) = 0$,
\begin{equation*}
I(X^T_0;\Param^0) = \frac{1}{2} \int^T_0 \E_{X^t_0}  \{ \Var_{\Param^0|X^t_0 } (F(x(t);\Param^0)) \} \de t.
\end{equation*}
For the general case where $x(0)$ might depend
on $\Param^0$ (if, for example, $x(0)$ is the stationary
state of the system) we can write
$I(X^T_0;\Param^0) = I(x(0);\Param^0) + I(X^T_0;\Param^0 | x(0))$
and apply the previous identity to $I(X^T_0;\Param^0 | x(0))$.
Taking into account that $I(\hM_T(X^T_0));M(\Param^0)) \leq I(X^T_0;\Param^0)$
and making use of Fano's inequality $I(\hM_T(X^T_0));M(\Param^0))
\geq H(M(\Param^0)) - 1 - (\prob(\hM_T(X^T_0) \neq M(\Param^0) ))\log(|\mathcal{M}|)$ the results follows.
\end{proof}

The bound in Theorem \ref{th:main_lbound} is often too complex to be
evaluated.  Instead, the following corollary provides a more easily
computable bound for the case when $X^T_0$ is a stationary process.
\begin{corollary}\label{th:main_simpler_lbound}
Assume that \eqref{eq:BasicModel} has a stationary distribution
for every realization of $\Param^0$
and let $X^T_0$ be a trajectory following any such stationary distribution for
a specific realization of the random variable $\Param^0$. Let $\hM_T(X^T_0)$ be an estimator of
$M(\Param^0)$ based on $X^T_0$.
If $\prob_{\Param^0,X^T_0}(\hM_T(X^T_0) \neq M(\Param^0) ) <\frac{1}{2}$ then
\begin{equation} \label{eq:main_simpler_bound}
T \geq \frac{2 H(M(\Param^0)) - \log (\mathcal{|M|})- 2 I(\Param^0;x(0)) -2}{ \E_{x(0)}
\{ \Var_{\Param^0|x(0) } (F(x(0);\Param^0)) \}},
\end{equation}
where $|\mathcal{M}|$ is the size of the alphabet of $M(\Theta^0)$.
\end{corollary}
\begin{proof}
Since conditioning reduces variance, we have 
\begin{align*}
\E_{X^t_0}  \{
&\Var_{\Param^0|X^t_0 } (F(x(t);\Param^0)) \} \\
& \quad \leq \E_{x(t)}  \{ \Var_{\Param^0|x(t) }
(F(x(t);\Param^0)) \}. 
\end{align*}
Using stationarity, we have 
\begin{align*}
&\E_{x(t)}  \{ \Var_{\Param^0|x(t) } (F(x(t);\Param^0)) \} \\
& \quad =\E_{x(0)}  \{ \Var_{\Param^0|x(0) } (F(x(0);\Param^0)) \}, 
\end{align*}
which simplifies \eqref{eq:main_bound} to \eqref{eq:main_simpler_bound}.
\end{proof}
In the rest of section \ref{sec:extensions}, we apply this lower bound to special
classes of SDEs, namely linear SDEs with dense matrices $\Param^0$ and
non-linear SDEs. In all of our applications it is to be understood that the 
process $\{x_{t}\}_{t\ge 0}$ is stationary.

%
%*********************************************************
%

\subsection{Learning dense linear SDEs}
\label{sec:linear_dense_lower_bound}

A different regime of interest in learning the
network of interactions for a linear SDE is the case
of dense matrices. As we shall see shortly, this regime
exhibits fundamentally different behavior in terms
of sample-complexity compared to the regime of sparse matrices. 

Let $\cA^{(D)} \subset \reals^{p\times p}$ be the
set of dense matrices defined as $\Param \in \cA^{(D)} $ if and only if, 
\begin{itemize}
\item[(i)] $ \param_{\min}\leq |\param_{ij}| p^{1/2} \leq \param_{\max} \forall i,j: \param_{ij} \neq 0$,
\item[(ii)] $\lambda_{\min}(-(\Param+\Param^*)/2) \ge \rho_{\min} > 0$.
\end{itemize}

The following theorem provides a lower bound for learning the signed
support of models from the class $\cA^{(D)}$ from stationary
trajectories $X^T_0$ of \eqref{eq:BasicModelLin}.%

\begin{theorem}\label{th:linear_lbound_dense}
Let $\alg = \alg(X^T_0)$ be an estimator of $\sign(\Param^0)$.
There is a constant $C(\delta)$
such that, for all $p$ large enough,
\begin{align}
T_{\alg}(\cA^{(D)}) \geq  C(\delta) \max \Big
\{\frac{\rho_{\min}}{\param^2_{\min}},
\frac{1}{\param_{\min}} \Big\} p.
\label{eq:lbound_dense}
\end{align}
\end{theorem}

The sample-complexity bound is similar to the one in Theorem \ref{th:linear_lbound_sparse}
but the scaling with $p$ has now changed from $\bigo(\log p)$ to $O(p)$.
The lack of structure in $\Param^0$ requires exponentially more
samples for successful reconstruction.
The proof is deferred to Section \ref{sec:appendix_lbound_dense_proof} in the appendix.

\begin{remark} \label{sec:remark_dense_linear_sdes}
Although the above theorem only gives a lower bound
on $T_{\rls(\lambda)}(\cA^{(D)})$, it is not hard to upper
bound $T_{\rls(\lambda)}(\cA^{(D)})$ for linear dense systems of SDEs
and certain values of $\lambda$.
In particular, it is not hard to upper
bound $T_{\rls(\lambda = 0)}(\cA^{(D)})$ by $\bigo(p)$. 	
This can be done in two steps.
First, taking $\lambda = 0$, one can
compute a closed form solution for $\rls$.
This solution is an unbiased estimator
involving sums of dependent Gaussian random variables.
Second, one can prove concentrations bounds similar
to the ones proved for Theorem \ref{th:main_cont_time},
and compute the trajectory length $T$
required to guarantee that
\begin{equation} \label{eq:remark_dense_sdes}
\|\hParam - \Param^0\|_{\infty} \leq \param_{\min}/2
\end{equation}
with probability greater than $1 - \delta$.
This value of $T$ is an upper bound on
$T_{\rls(0)}(\cA^{(D)})$ since
\eqref{eq:remark_dense_sdes} plus a simple thresholding
decision rule
\footnote{If $|\hat{\param}_{ij}|< \theta_{\min}/2$ declare
$0$, if $\hat{\param}_{ij} < -\theta_{\min}/2$ declare $-1$
and if $\hat{\param}_{ij} > \theta_{\min}/2$ declare $+1$.} is
enough to guarantee that
\begin{equation}		
\sign(\hParam) = \sign(\Param^0).
\end{equation}
We start Section \ref{sec:num_ext} with a numerical
illustration of this behaviour.
\end{remark}

%
%*********************************************************
%

\subsection{Learning (sparse) non-Linear SDEs}
\label{sec:non_linear_lower_bound}

We now assume that the observed samples $X^T_0$ come from
a stochastic process driven by a general SDE of the form \eqref{eq:BasicModel}.

In what follows, $v_i$ denotes the $i^{th}$ component of vector $v$.
For example, $x_3(2)$ is the $3^{th}$ component of the
vector $x(t)$ at time $t = 2$. $JF(\,\cdot\, ;\Param^0)\in\reals^{p\times p}$
denotes the Jacobian of the 
function $F(\,\cdot\, ;\Param^0)$.

For fixed $L$, $B$ and $D\ge 0$, 
define the class of functions 
$\cA^{(N)}=\cA^{(N)}(L,B,D)$ by letting $F(x;\Param) \in \cA^{(N)}$ if and only if
\begin{itemize}
\item[(i)] the support of $JF(x;\Param)$ has at most $k$ non-zero entries for every $x$,
\item[(ii)] the SDE \eqref{eq:BasicModel} admits a
stationary solution with covariance matrix, $Q$,
satisfying $\lambda_{\min}(Q) \geq L$,
\item[(iii)] $\Var_{x(0)|\Param}(x_i(0)) \leq B \; \forall i$,
\item[(iv)] $|\partial F_i(x;\Param) / \partial x_j |\leq D$ for all
$x\in\reals^p$ $i,j\in [p]$.
\end{itemize}

For simplicity we write $F(x;\Param^0) \in \cA^{(N)}$ by $\Param^0 \in \cA^{(N)}$.

Note that our objective is different than before.
Given $\Param^0 \in \cA^{(N)}$,
we are interested in recovering the smallest
support, $M(\Param^0)$,
for which $\supp(JF(x;\Param^0))
\subseteq M(\Param^0) \; \forall x$.
Hence, we consider the following modified definition
of sample-complexity
that can be applied to learning SDEs of the
form \eqref{eq:BasicModel},
\begin{align*}
& T_{\alg}(\cA^{(N)}) =
\sup_{\Param^0 \in \cA^{(N)}} \inf \big\{T_0 \in \reals^+: \prob_{\Param^0,T}\{\alg(X^T_0) \\
& \qquad = M(\Param^0)\}  \geq 1- \delta \text{ for all } T \geq T_0 \big\}.
\end{align*}

The following theorem holds for learning $M(\Param^0)$, $\Param^0 \in \cA^{(N)}$,
from a stationary trajectory of \eqref{eq:BasicModel}.

\begin{theorem} \label{th:non_linear_lbound}
Let $\alg = \alg(X^T_0)$ be an estimator of $M(\Param^0)$.
Then  
\begin{equation}
T_{\alg}(\cA^{(D)}) \geq \frac{k \log p/k
- \log B/L }{C + 2 k^2 D^2 B},
\end{equation}
where
$C = \max_{i \in [p]}
\E \{F_i(\E_{x(0)|\Param^0}(x(0));\Param^0)\}$.
\end{theorem}

\begin{remark}
The assumption that $F$ is Lipschitz is not very restrictive as it is a sufficient condition commonly used to guarantee existence and uniqueness of a solution of the SDE \eqref{eq:BasicModel} with finite expected energy, \cite{oksendal2003stochastic}.
\end{remark}

%
%*********************************************************
%

\section{Numerical illustration of some extensions} \label{sec:num_ext}

In Theorem \ref{th:main_cont_time} we describe a set
of conditions under which $\rls$ successfully
reconstructs the dynamics of a sparse system of linear SDEs.
These sufficient conditions naturally raise several
questions: do they hold when the entries
of $\Param^0$ are related to
some real world problem? Can $\rls$ perform well
even when these conditions do not hold?
Even more generally, can $\rls$ learn SDEs in a
scenario completly different than the one described
in Theorem \ref{th:main_cont_time}, e.g. in the presense of non-linearities?
Answering these questions is non-trivial
because it is hard to get a clear intuition of what assumptions
like Assumption 1 and Assumption 2 of Section \ref{sec:SparseLinear}
mean in practice. The same
difficulty arises with analogous results on the high-dimensional
consistency of the LASSO method\cite{wainwright2007high, zhao}.

In this section we provide concrete illustrations of the performance of $\rls$ when applied to scenarios
for which our upper bounds on time-complexity do not hold.
We compare its performance to the performance
predicted by our lower bounds, that hold in greater generality,
and observe that, in these examples, they match. Finally, although not the focus of this paper, our last example also illustrates the effect of $\lambda$ on the performance of
$\rls(\lambda)$.

%
%*********************************************************
%

\subsection{Time-complexity for dense linear SDEs}

First we study the time-complexity for learning dense linear SDEs.
We repeat the experiment of Section \ref{sec:num_resul_main_th}
for continuous-time linear SDEs but with a
dense matrix $\Param^0$ that we generate as  follows.
%
%\vspace{-2.0em}
%
\begin{figure}[H]
\begin{center}
\includegraphics[trim=3cm 7.5cm 3.5cm 8cm,clip,width = .31 \textwidth]{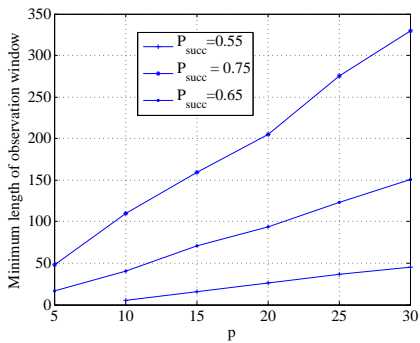}
\caption{Time-complexity to learn linear dense
SDEs as a function of the dimension of $\Param^0$ for different
probabilities of success. }
\end{center}
\label{fig:samp_comple_linear_dense}
\end{figure}

Generate $\tilde{\Param}^0 \in \reals^{p \times p}$
by sampling each entry from a standard Gaussian distribution; set every entry to zero with probability $1/2$; set  $\Param^0 = -(\rho + \sqrt{2} )\mathbb{I} + p^{-1/2} \tilde{\Param}^0$.
For large $p$, almost all such generated
matrices lead to a stable SDE.

The time-complexity curves we obtain are depicted in the
figure above. Just like pointed out in Remark \ref{sec:remark_dense_linear_sdes},
we observe that the time-complexity scales linearly with $p$, 
compared to $\bigo( \log p)$ for sparse matrices. The slope
is larger for larger probabilities of success.

\subsection{Time-complexity for non-linear SDEs}

The example in this subsection illustrates
that the time-complexity of $\rls$ scales
like $\bigo(\log p)$ even when learning sparse non-linear systems
of SDEs.

Consider a system of $p$ masses  in
$\reals^d$ connected by damped springs
that is vibrating under the influence of white-noise.
These can be thought of, for example, as points 
on a vibrating object whose physical structure we
are trying to reconstruct from the measured amplitude
of vibrations over time on a grid of points at its surface.

Let $C^0$ be the corresponding adjacency matrix,
i.e. $C^0_{ij}=1$ if and only if masses $i$  and $j$ are connected, and
$D^0_{ij}$  be the rest length of the spring $(i,j)$. Assuming unit
masses, unit rest lengths and unit elastic coefficients,
the dynamics of this system in
the presence of external noisy forces 
can be modeled by the following damped Newton equations
\begin{align} 
&\de v(t) = -\gamma v(t) \de t-\nabla U(q(t))\, \de t  +\sigma\, \de b(t), \label{eq:general_ms_dynamics}\\
&\de q(t) = v(t) \de t\, ,\label{eq:general_ms_dynamics_2}\\
& U(q) \equiv
\frac{1}{2}\sum_{(i,j)}C^0_{ij}(\|q_i-q_j\|-D^0_{ij})^2\, ,\nonumber
\end{align}
where $q(t) =
(q_1(t),\dots,q_p(t))$, $v(t) = (v_1(t),\dots,v_p(t))$, and
$q_i(t), v_i(t) \in \reals^d$ denote the position and velocity of mass $i$ at time $t$. 
This system of SDEs can be written in the form \eqref{eq:BasicModel}
by letting $x(t) = [q(t), v(t)]$ and $\Param^0 =[C^0, D^0]$.
A straightforward calculation shows that 
the drift $F(x(t);\Param^0)$ can be further written as a linear combination
of the following basis of non-linear functions
\begin{align*}
\bF(x(t)) = \Big[& \{v_i(t)\}_{i \in [p]}, \{ \Delta_{ij}(t)\}_{i,j \in [p]}, \\
& \Big\{ \frac{\Delta_{ij}(t)}{\|\Delta_{ij}(t)\|} \Big\}_{i,j \in [p]} \Big],
\end{align*}
where $\Delta_{ij}(t) = q_i(t)-q_j(t)$ and $[p] =
\{1,\dots,p\}$. Hence, the system can be modeled
according to \eqref{eq:BasicModelBasis}.
In many situations, only specific properties of the
parameters are of interest, for instance one might be interested
only in the network structure of the springs. 

We consider the trajectories of three
masses in a two-dimensional network of 36 masses and 90 springs evolving according
to Eq.~(\ref{eq:general_ms_dynamics}) and Eq.~(\ref{eq:general_ms_dynamics_2}).
How long does one need to observe these (and the other masses) trajectories 
in order to learn the structure of the underlying network? 
Notice that the system being considered is non-linear
and hence, a priori, we cannot apply any of our theorems
to guarantee that correct reconstruction will be achieved
for any $T$.
%
%\vspace{-2.0em}
%
%\begin{figure}[H]
%\begin{center}
%\includegraphics[trim=0.86cm 0.56cm 0cm 0cm,clip,width = 2.1in]{Figures/trace_1000seg}
%\put(-48,5){Time (seconds)}
%\put(-225,-10){\tiny $0$}
%\put(-193,-10){\tiny $200$}
%\put(-153,-10){\tiny $400$}
%\put(-110, -10){\tiny $600$}
%\put(-70,-10){\tiny $800$}
%\put(-30,-10){\tiny $1000$}
%\put(-210,152){Displacement (meters)}
%\put(-233,160){\tiny $9$}
%\put(-233,139){\tiny $8$}
%\put(-233,115){\tiny $7$}
%\put(-233,92){\tiny $6$}
%\put(-233,68){\tiny $5$}
%\put(-233,45){\tiny $4$}
%\put(-233,22){\tiny $3$}
%\put(-233,0){\tiny $2$}
%\end{center}
%\caption{Evolution of the horizontal component of the position of
 % three masses in a system with $p=36$ masses interacting via elastic
 % springs (cf. Fig.~\ref{fig:ms_network_reconstrunct} for the network
%  structure).  The time interval is $T=1000$. All the springs have
%  rest length $D_{ij}=1$, the damping coefficient is $\gamma=2$,
% cf. Eq.~(\ref{eq:general_ms_dynamics}), 
%and the noise variance is $\sigma^2 =0.25$.}\%label{fig:orig_ms_network}
%\end{figure}
%
%\vspace{-1.0em}
%
Figure \ref{fig:ms_network_reconstrunct} reproduces the network structure 
reconstructed using the ${\sf RLS}$ algorithm described in Sec. \ref{sec:ResContinuumRegularized}
for increasing observation intervals $T$. 
%
%\vspace{-2.0em}
%
\begin{figure}[H]
\begin{center}
%\put(10,10){$T = 500$}
\includegraphics[scale = 1]{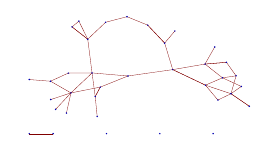}
\includegraphics[scale = 1]{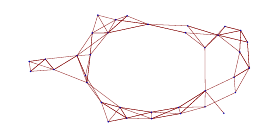}
\includegraphics[scale = 1]{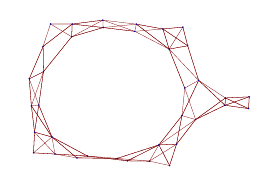}
\includegraphics[scale = 1]{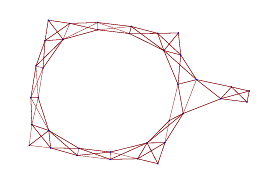}
\includegraphics[scale = 1]{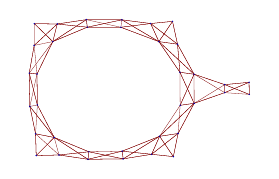}
\caption{From left to right, top to bottom: structures  reconstructed using $\rls$ with 
 observation time $T=500$, $1500$, $2500$, 
$3500$ and $4500$. For $T=4500$ exact reconstruction is achieved.}\label{fig:ms_network_reconstrunct}
\end{center}
\end{figure}
%
%\vspace{-2.0em}
%
Despite
the non-linearities, the inferred structure
converges to the true one when $T$ is large enough
\footnote{The data was generated by a simulation of
Newton's equations of motion using an Euler
approximation with discrete time step of size $0.1$s}.

To quantify the efficiency of the regularized least-squares
in learning non-linear SDEs,
we generate multiple spring-mass networks
of sizes $p = 8, 16, 32, 64$ and $128$ and study the
mean minimum length of the observation window required for successful
reconstruction. The spring-mass networks
are sampled uniformly from the ensemble of regular graphs
of vertex degree 4. Like for the previous system,
the data is generated by simulating the dynamics
using an Euler approximation with a time step of $0.1$s.
The noise level, $\sigma$, is set to $0.5$ and the
damping parameter, $\gamma$, is set to $0.1$.

Figure \ref{fig:mass_spring_more_detail}--top
shows the probability
of success versus the length of the observation time window 
for systems of different sizes ($p = 8, 16, 32, 64$ and $128$)
and Figure \ref{fig:mass_spring_more_detail}--bottom
shows the minimum length of observation window for successful
reconstruction of the networks versus their size for
different probabilities of success ($\psucc = 0.1, 0.5$ and $0.9$).
In both pictures, error bars represent $\pm$
two standard errors. 
We define a successful reconstruction by an exact recovery of the whole network.
Since networks are sampled uniformly over 
regular graphs, the probability
of full exact reconstruction of the
network equals the probability of full exact
reconstruction of any node's neighborhood in the network.
This fact is used to minimize the number of
simulations required to achieve a small
fluctuation in our numerical results.

\begin{figure}[]
\begin{center}
\includegraphics[trim=2.5cm 6.5cm 3.3cm 7cm,clip,width = 2.25in]{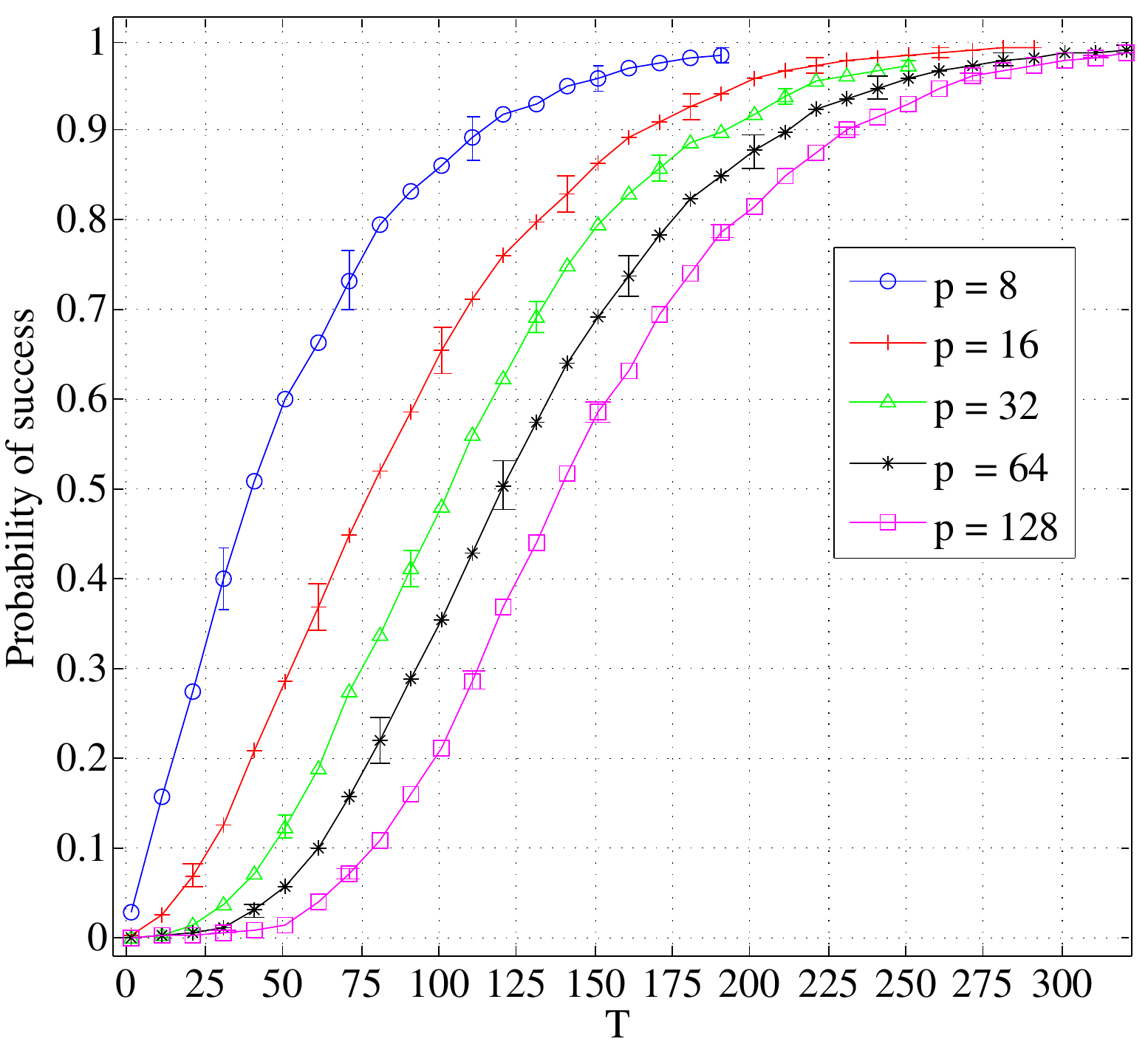}\\
\includegraphics[trim=1.2cm 6.3cm 2cm 6cm,clip,width = 2.60in]{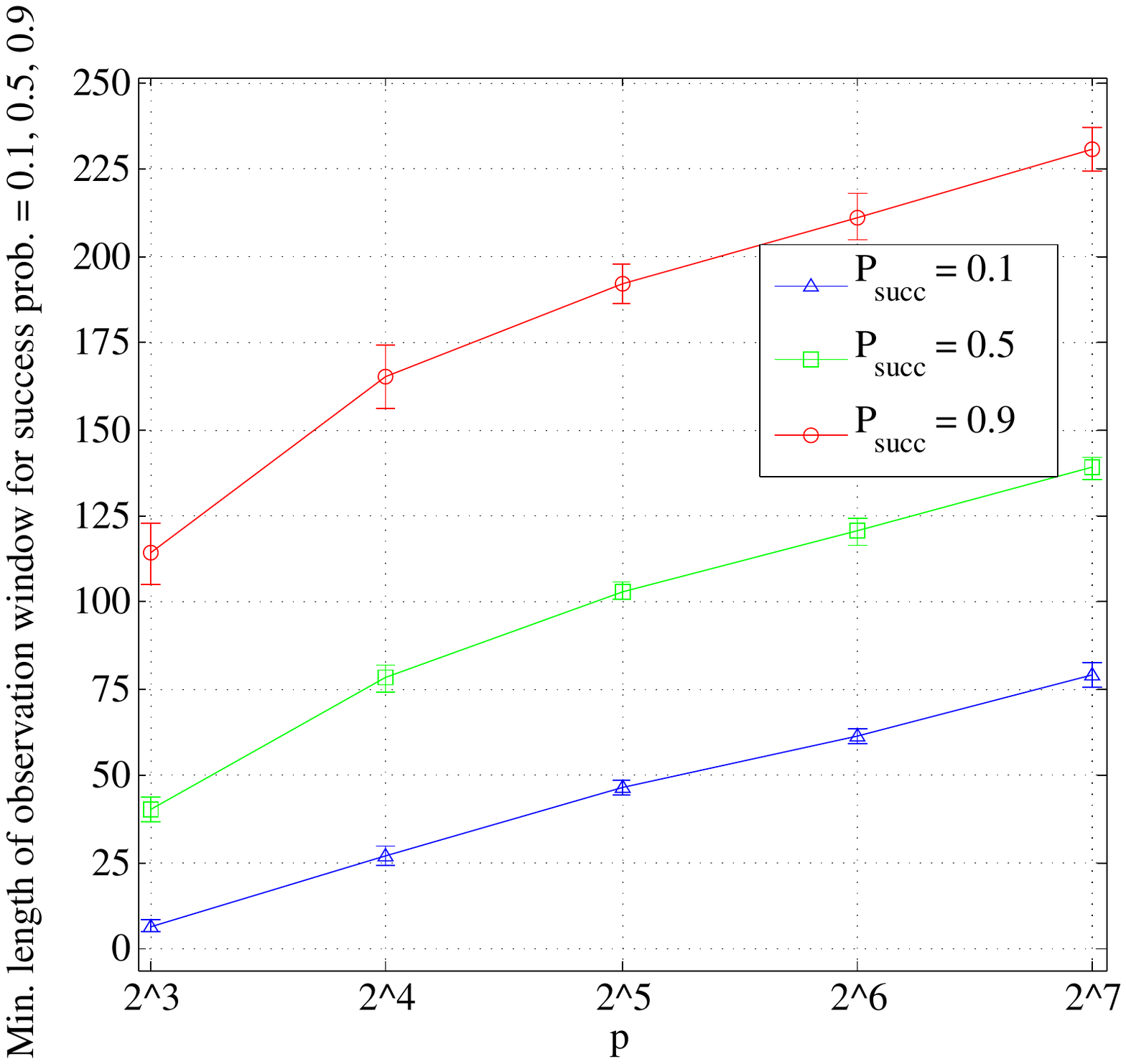}
\caption{(top) Probability of success versus length of
observation time window, $T$, for different network sizes
($p = 8, 16,32, 64$ and $128$).
(bottom) Minimum number of samples required to achieve
a probability of reconstruction of $\psucc = 0.1, 0.5$ and $0.9$
versus the size of the network $p$. All networks
where generated from random regular graphs
of degree 4 sampled uniformly at random.
The dynamics' parameters were set to
$\sigma = 0.5$ and $\gamma = 0.1$}
\label{fig:mass_spring_more_detail}
\end{center}
\end{figure}

In agreement with the lower bound of Theorem \ref{th:non_linear_lbound} for non-linear SDEs,
the time-complexity of $\rls$
in learning these sparse non-linear system of SDEs
also scales logarithmically with $p$.
The behavior of the plot also agrees with the
$\bigo(\log p)$ time-complexity for sparse linear SDEs,
even though the mass-spring system is non-linear (compare
Figure \ref{fig:mass_spring_more_detail}
with Figure \ref{fig:psVsN}). A careful
look into the proof of our main theorem suggests
that as long as the correlation
between consecutive samples decays exponentially
with time, the same proof should follow despite the non-linearities.
The difficulty in proving a generalization
of Theorem \ref{th:main_cont_time} to general non-linear SDEs
of the from \eqref{eq:BasicModelBasis} stems from the fact that
it is hard in to know what kind
of correlations a general SDE will induce on its
trajectory. However, given a sufficiently 'nice' trajectory
the success of the least-square method should not
be affected by the fact that we are considering
a non-linear basis of functions. In fact, even in this
case, the method still consists of minimizing a quadratic
function under a norm-1 constrain.

%
%************************************************
%

\subsection{Learning biochemical pathways and
the effect of the regularization parameter}

We now look at a biochemical pathway describing a general response of a cell to a change in its environment. We model the pathway behavior using non-linear SDEs, produce synthetic data by simulation and then try to recover it from the data using $\rls$. In this example, we also analyze how the regularization parameter, $\lambda$, affects the support recovery and the (normalized) error in estimating the values of $\Param^0$.

The pathway in consideration is described in \cite{aldridge2006physicochemical} and reproduced below.
\begin{align}
R+L &\xrightleftharpoons[k_{r1}]{k_{f1}} (LR^*),\\
(LR^*)+K &\xrightleftharpoons[k_{r2}]{k_{f2}} (LR^*K),\\
(LR^*K) &\xrightarrow[]{k_{f3}} (LR^*) + K^*,\\
K^* + S &\xrightleftharpoons[k_{r4}]{k_{f4}} (K^*S),\\
(K^*S) &\xrightarrow[]{k_{f5}} K^* + S^*.
\end{align}
This pathway can describe, for example, the response of cells to a lesion on the
skin. The lesion causes some cells to generate diffusible ligands ($L$). These
ligands come upon receptors ($R$) on the cell membrane, which act like antennas.
Receptors that have caught a ligand can then be modified (phosphorylated $*$) by
enzymes called kinases ($K$).  These modifications enable interactions with
other substrates ($S$) which eventually turn on the genetic program of platelets
to move towards the source of the injury. This sequence of events is what is
called a chemical pathway and can be thought of as a sequence of chemical reactions describing the interaction between difference species inside and outside the cell.
The symbols, $k_f$ and $k_r$ are the forward and backward rates of reaction.
Expressions inside parenthesis, e.g. $(LR^*)$, represent specific intermediary stages or compounds along the pathway. 

We assume the following
correspondence between the concentration of each
species and the variables $x_i(t), i \in [9]$:
$x_1 \leftrightarrow R$, $x_2 \leftrightarrow L$, $x_3 \leftrightarrow (LR^*)$,
$x_4 \leftrightarrow (LR^*K)$, $x_5 \leftrightarrow K$, $x_6
\leftrightarrow K^*$, $x_7 \leftrightarrow S$, $x_8 \leftrightarrow (K^*S)$,
$x_9 \leftrightarrow (S^*)$. With this notation, the model
proposed in \cite{aldridge2006physicochemical} takes the form of a
system of non-linear SDEs. Bellow are a few of the equations in the model.
\begin{align*}\label{eq:nine_gene_sde}
\de x_1(t) &= (k_{r1}x_3(t) - k_{f1}x_1(t)x_2(t))\de t + \de b_1(t) \nonumber \\
\de x_2(t) &= (k_{r1}x_3(t) - k_{f1}x_1(t)x_2(t))\de t + \de b_2(t) \nonumber \\
...\\
\de x_8(t) &= (k_{f4}x_6(t)x_7(t) - (k_{r4}+k_{f5})x_8(t))\de t + \de b_8(t) \nonumber \\
\de x_9(t) &= (k_{f5}x_8(t) )\de t + \de b_9(t) \nonumber
\end{align*}
The data we use for learning are synthetic sample-trajectories for the concentrations, $\{x_i(t)\}^9_{i=1,t\geq 0}$,  obtained
from these equations using the Euler-Maruyama method.

We learn the network of interaction as the support of a 
non-linear SDE of the form \eqref{eq:BasicModelBasis}
with a basis of functions consisting of monomials up to order two,
i.e., all the functions of the form $x_i^{\alpha_1}x_j^{\alpha_2}$
with $\alpha_1, \alpha_2 \in \{0,1\}$. 
Although there are only 9 species in the model, the adjacency matrix whose support we want to learn is $\Param^0 \in \mathbb{R}^{9\times 46}$ which translates into $414$ parameters to be estimated.

Figure \ref{fig:nine_gene_net_roc} summarizes the performance of $\rls$ in recovering the support
of $\Param^0$. Figure \ref{fig:nine_gene_net_roc}-top
shows that, for a fixed value of $\lambda$, as the length of the observation increases from $T = 150$ to $T = 3000$, in steps of $285$, the number of species that do not interact that are estimated as interacting (false positives) decreases and the number of species that do interact that are estimated as interacting (true positives) increases. It also shows that one can go from a high true positive rate and a high false positive rate 
to a low true positive rate and a low false positive rate by increasing $\lambda$. Figure \ref{fig:nine_gene_net_roc}-bottom shows the area under the previous curves as a function of $T$. In this case, the area under the curve does not have the usual probabilistic interpretation, but it does provide a metric of performance for $\rls$ that is independent of $\lambda$. The area increases with $T$ and approaches $1$, i.e. $\rls$ can recover the exact structure of the biochemical pathway if enough data is available.
%
%\vspace{-3.0em}
%
\begin{figure}[h]
\center 
\includegraphics[width=.4\textwidth]{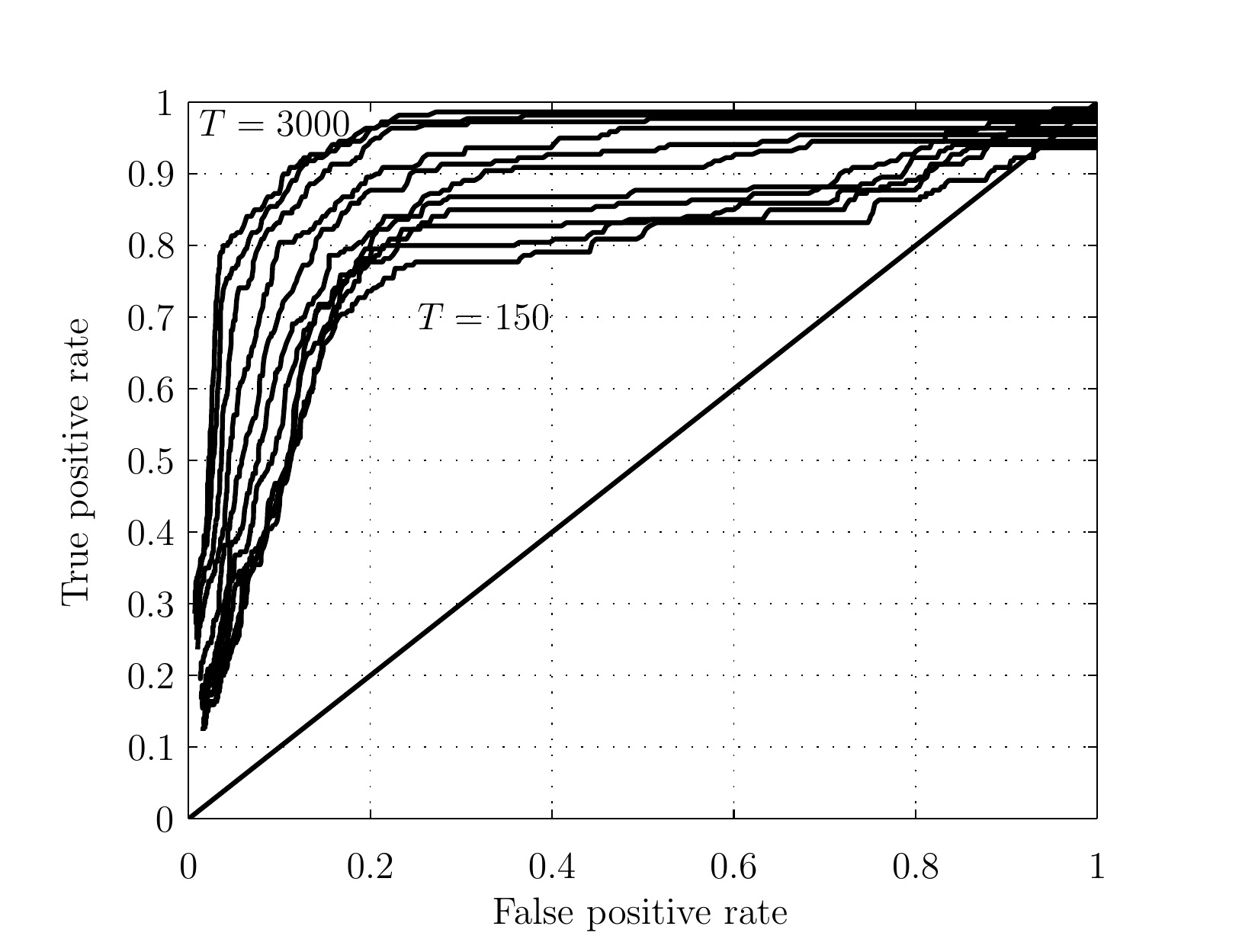}\\
\includegraphics[width=.4\textwidth]{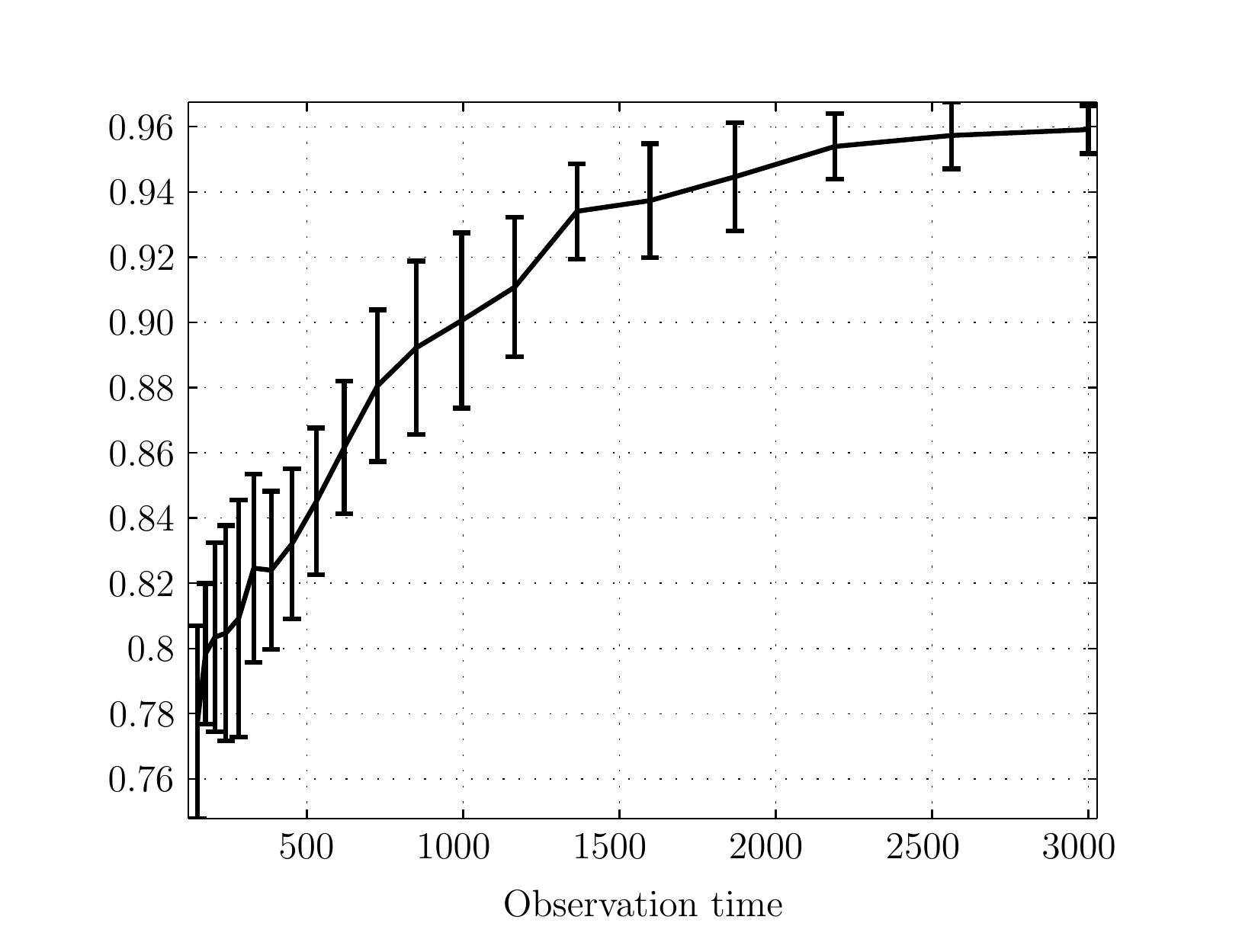}
\put(-200,57){\rotatebox{90}{\scriptsize Area under the curve}}
\caption{(top) True positive rate versus false positive rate for
the recover of the entries in the support of $\Param^0$ using $\rls$. The regularization parameter $\lambda$ changes along each curve. As $\lambda$ decreases (from $\infty$ to $0$), the true positive rate increases but the false positive rate also increases. (bottom) Area under the curves above as a function of $T$.}
\label{fig:nine_gene_net_roc}
\end{figure}

Although the focus of this paper is on support recovery,
$\rls$ also outputs real-value estimates for
the entries of $\Param^0$. Hence one can look
at the normalized RMSE $\III \hat{\Param} - \Param^0 \III_{\text{fro}} / \III\Param^0 \III_{\text{fro}}$ and
its relation with $\lambda$. Figure \ref{fig:nine_gene_net_rmse}-top shows this relation when running $\rls$ on $T = 1200$ seconds of data.
The curve follows the typical behavior described
in \cite{bayati2010lasso}. In particular,
there is a value of $\lambda$ that gives best 
parameter estimation. In Figure \ref{fig:nine_gene_net_rmse}-bottom we show the evolution of the value of the minimum normalized RMSE as a function of $T$ up to the
maximum duration we simulated.
\begin{figure}[h]
\center 
\includegraphics[width=.4\textwidth]{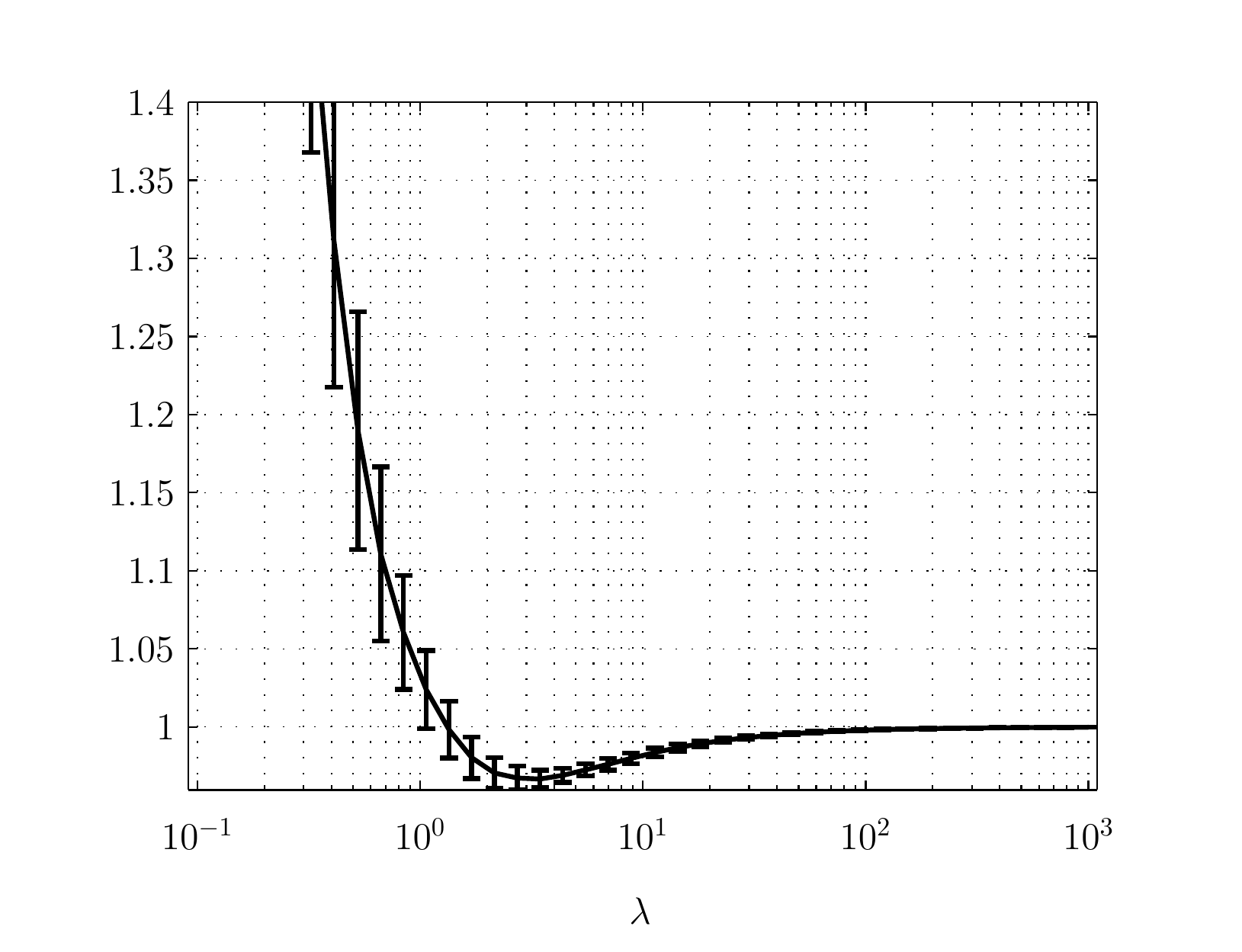}
\put(-200,56){\rotatebox{90}{\scriptsize Normalized RMSE}}
\\
\includegraphics[width=.4\textwidth]{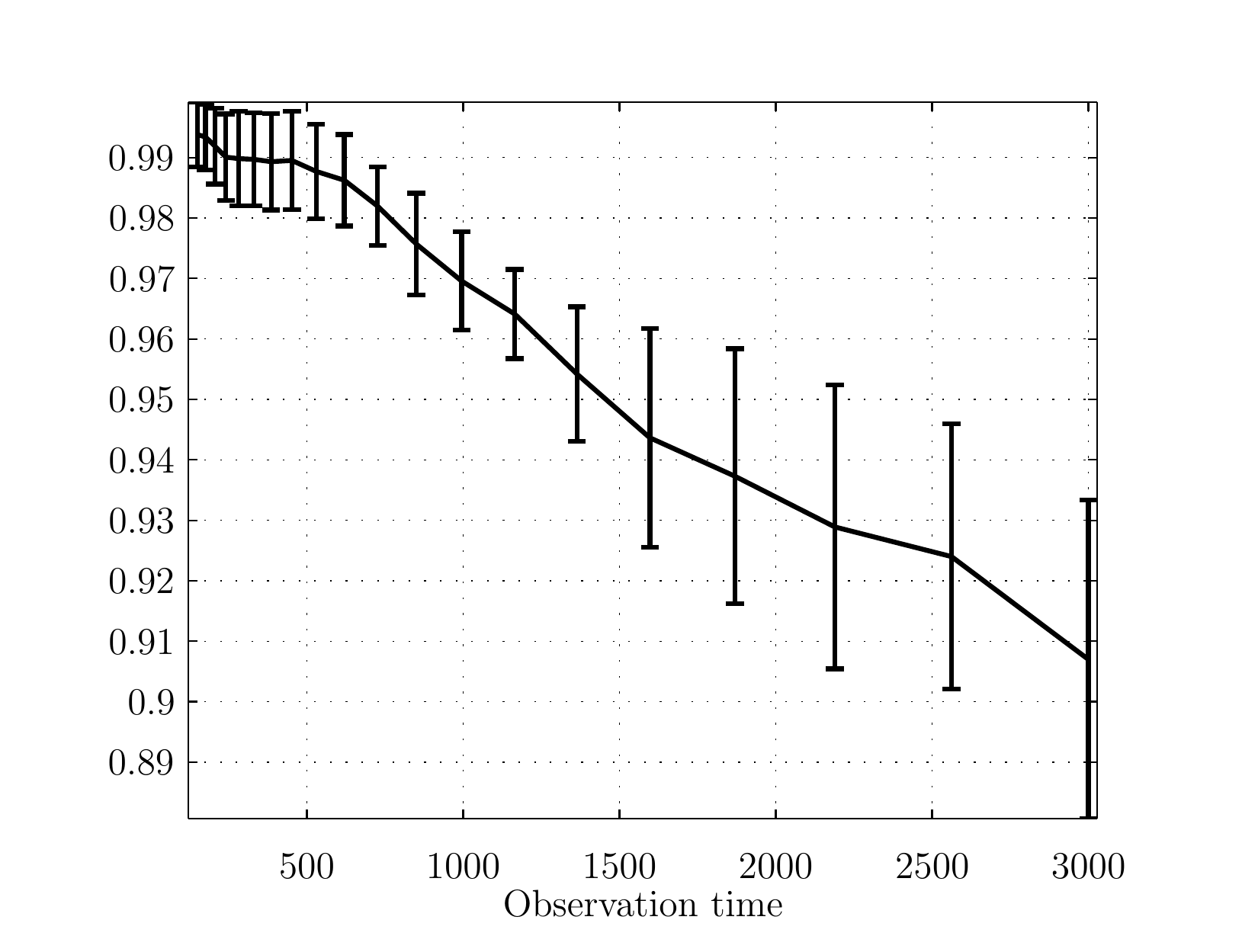}
\caption{(top) Normalized RMSE versus $\lambda$ for $T = 1200$. (bottom) Normalized RMSE versus $T$ for best value of $\lambda$.}
\label{fig:nine_gene_net_rmse}
\end{figure}
%

%
%*********************************************************
%

\subsection*{Acknowledgments}

This work was partially supported by the NSF CAREER award CCF-0743978,
the NSF grant DMS-0806211, the AFOSR grant FA9550-10-1-0360 and by a Portuguese Doctoral FCT fellowship. A great part of this work was
done under the suppervision of Professor
Andrea Montanari. We are very thankful for
his help and contribution. Finally,
we thank Jiantao Jiao 
and Peter Trocha  for their feedback
on this document.
%
%\vspace{-4.0em}
%
%*****************************************************
%

%
%*****************************************************
%

\bibliographystyle{IEEEbib}
\bibliography{inference}
\clearpage 

\begin{IEEEbiography}{Jose Bento Ayres Periera}
Biography text here.
\end{IEEEbiography}

\begin{IEEEbiography}{Morteza Ibrahimi}
Biography text here.
\end{IEEEbiography}

\newpage

%
%We now prove all of our results. 
%First however we give more details regarding
%the motivation example of Section \ref{sec:5_gene_network},
%the 5 gene network. Then we
%prove the upper bounds on the sample complexity and
%finally proceed to prove the lower bounds.
%Each of the proofs is included in a separate section. 
%Auxiliary Lemmas are introduced as needed and are 
%proven in a separate subsection inside these sections.

%
%*********************************************************
%
\appendices

The appendix is divided into two parts. The first half
contains the proofs of the upper bounds on the
sample-complexity and the second half the proofs of the
lower bounds.

\section{Proofs of the upper bounds on the sample-complexity of the regularized least square algorithm}

Our bounds for the continuous model follow from an analysis of the problem for discrete case (introduced in Section \ref{sec:DiscreteResults}) when taking the limit when $\eta \rightarrow 0$.
Hence, we first prove Theorem \ref{th:main_discrete}. We begin by giving an outline of the
proof in Section \ref{sec:outline_proof_discrete_case_upper_bound}
based on three propositions. The three propositions
are proved in Section \ref{sc:auxi_upper_bound}
and, in particular, the details of how to combine them
to complete the proof of Theorem \ref{th:main_discrete}
are in Section \ref{sc:full_main_disc_theo_proof}.
Afterwards, in Section \ref{sec:proof_cont_theorem_upper_bound_appendix}, we prove Theorem \ref{th:main_cont_time}.
Finally, in Section \ref{sec:proof_of_reg_graph_upper_bound}, we specialize this bound to the case of the Laplacian of a graph and prove Theorem \ref{th:cont_reg_graph_bounded_degree} .

%
%*********************************************************
%

\subsection{Proof of Theorem \ref{th:main_discrete}}
\label{sec:outline_proof_discrete_case_upper_bound}

In this Section we detail the proof of our main result for discrete-time
dynamics, i.e., Theorem \ref{th:main_discrete}. We start by stating
a set of sufficient conditions for regularized least squares to recover the correct support and sign of the entries of $\Param^0$.
Then we present a series of concentration lemmas to be used to prove the validity of these conditions, and then finalize the proof. 

As mentioned in the main text, the proof strategy, and in particular the following proposition, Proposition \ref{th:cond_to_hold}, which provides a compact set of sufficient conditions for the sign-support to be recovered correctly, is analogous to the one in 
\cite{zhao}. A proof of this proposition can be found in in Section \ref{sec:proof_of_norm_bounds_to_prove_for_cont_theorem}. %is deferred to Appendix \ref{pf:compact_condition}.

In the following we denote by $X\in\reals^{p\times n}$
the matrix whose $(t+1)^{\text{th}}$ column corresponds to the configuration $x(t)$,
i.e. $X = [x(0), x(1), \dots, x(n-1)]$. Furthermore,  $\Delta X \in \reals^{p\times n}$ 
is the matrix containing consecutive state changes, namely  
$\Delta X=[x(1)-x(0), \dots, x(n)-x(n-1)]$.
It is important not to confuse $X^n_0 \equiv \{x(t):   t-1 \in [n-1]\}$ with $X$ defined here.
These are not the same, although both are related. In addition, $X^n_0$ should not be confused with the $n^{th}$
power of $X$ (which is never mentioned in this paper).
Finally we write $W = [w(1),\dots,w(n-1)] \in \reals^{p \times n}$ for the matrix containing 
the Gaussian noise realization and observe that
\begin{align*}
W = \Delta X-\eta \Param\, X\, .
\end{align*}
The $r^{\rm th}$ row of $W$ is denoted by $W_r$.

In order to lighten the notation, we omit the reference to 
$X_0^n$ in the likelihood function (\ref{eq:opt_prob}) 
and simply write $\cL(\Param_r)$. 
We define its normalized gradient and Hessian by
\begin{align}
\hG & = - \nabla \cL (\Param^0_r) =  \frac{1}{n \eta} X W_r^* \, \text{   \;\;\;   and  \;\;\;    }\nonumber \\
\hQ &=  \nabla^2 \cL (\Param^0_r) =  \frac{1}{n} X X^*\, .
\end{align}

\begin{proposition}\label{th:cond_to_hold}
Let $\alpha, C_{\min}>0$ be  defined by
\begin{align}
\lambda_{\min}(Q^0_{{S^0},{S^0}}) & \equiv C_{\min} \nonumber \\
\III Q^0_{(S^0)^C, {S^0}} \left(Q^0_{{S^0},{S^0}}\right)^{-1} \III _\infty & \equiv
 1 - \alpha\, .
 \label{eq:cond_irrep}
\end{align}
If the following conditions hold then the
regularized least square solution \eqref{eq:discr_reg_prob} 
correctly recovers the 
signed-support of $\Param^0$, i.e. $\sign(\Param^0_r)$:
\begin{align}
\|\hG\|_\infty &\leq \frac{\lambda \alpha}{3} \label{eq:Conditions0}\, ,  \\
\|\hG_{S^0}\|_\infty & \leq \frac{\Param_{\min} C_{\min}}{4 k} - \lambda,\label{eq:Conditions1}\\
 \III \hQ_{(S^0)^C ,{S^0}} - Q^0_{(S^0)^C, {S^0}} \III _\infty &\leq
\frac{\alpha}{12} \frac{C_{\rm min}}{\sqrt{k}}\, ,  \\ 
 \III \hQ_{{S^0},{S^0}} - Q^0_{{S^0},{S^0}} \III _\infty & \leq \frac{\alpha}{12} \frac{C_{\min}}{\sqrt{k}}\, .\label{eq:matrix_norm_cond_12}
\end{align}
Further the same statement holds for the continuous model \ref{eq:ContCost},
provided  $\hG$ and $\hQ$ are the gradient and the Hessian of the 
likelihood (\ref{eq:ContCost}).
\end{proposition}

The proof of Theorem \ref{th:main_discrete}
consists in checking that, under the hypothesis 
\eqref{eq:sample_bound_disc} on the number of consecutive configurations,
conditions (\ref{eq:Conditions1}) to (\ref{eq:matrix_norm_cond_12}) hold with high probability.
Checking these conditions can be regarded in turn as 
concentration-of-measure statement. Indeed, if expectation is taken with 
respect to a stationary trajectory, we have
$\E\{\hG\}=0$, $\E\{\hQ\}= Q^0$. 
%
%***********************************
%

\subsubsection{Technical lemmas}

In this section we state the necessary concentration lemmas for proving Theorem \ref{th:main_discrete}. These are non-trivial because $\hG$, $\hQ$ are quadratic  
functions of \emph{dependent} random variables $\big($the samples $\{x(t)\}_{0 \le t\le n}$$\big)$. 
The proofs of Proposition \ref{th:gradie_prob_bound} and Proposition \ref{en:bddQij} can be found in Section \ref{sec:proof_of_concentration_bounds}.

Our first Proposition implies concentration of $\hG$ around $0$.
\begin{proposition} \label{th:gradie_prob_bound}
Let $S\subseteq [p]$ be any set of vertices and $\epsilon<1/2$.
If $\sigma_{\rm max}\equiv\sigma_{\max}(I+\eta\, \Param^0)<1$, then 
\begin{align}
\prob\big\{\|\hG_S\|_\infty > \epsilon\big\} \leq 2 |S|\, \exp\left(-n (1-\sigma_{\max})\, \epsilon^2/4 \right).
\end{align}
\end{proposition}
Furthermore, we need to bound the matrix norms as per \eqref{eq:matrix_norm_cond_12}  in proposition $\ref{th:cond_to_hold}$. First we relate bounds on $ \III \hQ_{JS} - {Q^0}_{JS}  \III_\infty$ with bounds on $|\hQ_{ij} - Q^0_{ij}|$, ($i \in J, j \in S$) where $J$ and $S$ are any subsets of $\{1,...,p\}$. Namely, we have
%\E(\hQ_{JS}(n+m+1)
\begin{align}\label{eq:UnionBound}
&\prob ( \III  \hQ_{JS} - Q^0_{JS})  \III _{\infty} > \epsilon ) \nonumber \\
&\qquad \leq |J| |S| \max_{i \in J, j \in S} \prob ( | \hQ_{ij} - Q^0_{ij} | > \epsilon / |S|).
\end{align}
Then, we bound $|\hQ_{ij} - Q^0_{ij}|$ using the following proposition
\begin{proposition}
Let $i,j\in \{ 1,...,p\}$,  $\sigma_{\max} \equiv \sigma_{max}(I + \eta \Param^0) < 1$, $T = \eta n > 3/D$ and $0 < \epsilon< 2/D$ where $D = (1-\sigma_{\max})/\eta$ then, 
\begin{align}
\prob(|\hQ_{ij} - Q^0_{ij})|  > \epsilon) \leq   2 \exp\left(-\frac{n}{32 \eta^2}  (1-\sigma_{\max})^3 \epsilon^2 \right) .
\end{align} \label{en:bddQij}
\end{proposition}
Finally,  the next corollary follows
from Proposition \ref{en:bddQij} and Eq.~\eqref{eq:UnionBound}.
\begin{corollary}\label{en:bddQJS}
Let $J,S$ ($|S| \leq k$) be any two subsets of $\{ 1,...,p\}$ and  $\sigma_{\max} \equiv \sigma_{\max}(I + \eta \Param^0)<1$, $\epsilon< 2 k/D$ and $n\eta > 3/D$ (where $D = (1 - \sigma_{\max})/\eta )$ then, 
\begin{align}
& \prob( \III \hQ_{JS} - Q^0_{JS} \III _{\infty}  > \epsilon) \nonumber \\
 & \qquad \leq   2 |J| k \exp\left(-\frac{n}{32 k^2 \eta^2} (1-\sigma_{\max})^3\epsilon^2\right).
\end{align} \label{th:prob_bound_inft_matrix_hess}
\end{corollary}
%
%**************************************
%
\subsubsection{Outline of the proof of Theorem \ref{th:main_discrete}}

With these concentration bounds we can now easily prove Theorem \ref{th:main_discrete}. All we need to do is to compute the probability that the conditions given by Proposition \ref{th:cond_to_hold} hold.
From the statement of the theorem we have that the first two conditions ($\alpha, C_{\min} > 0$) of Proposition \ref{th:cond_to_hold} hold. In order to make the first condition on $\hG$ imply the second condition on $\hG$, we assume that 
$\lambda \alpha/3 \leq (\param_{\min} C_{\min})/(4 k) - \lambda$
which is guaranteed to hold if
\begin{equation}
\lambda \leq \param_{\min} C_{\min} / 8k. \label{eq:th_main_disc_lambda} 
\end{equation}
We also combine the two last conditions on $\hQ$, thus obtaining the following sufficient condition
\begin{equation}
 \III \hQ_{[p], {S^0}} - Q^0_{[p], {S^0}} \III _\infty \leq \frac{\alpha}{12} \frac{C_{\min}}{\sqrt{k}}\, , 
\end{equation}
since $[p]=S^0 \cup (S^0)^C$. We then impose that both the probability of the condition on $\hQ$ failing and the probability of the condition on $\hG$ failing are upper bounded by $\delta/2$ using Proposition \ref{th:gradie_prob_bound} and Corollary \ref{th:prob_bound_inft_matrix_hess}. 
It is shown in the end of this section, Section \ref{sc:full_main_disc_theo_proof}, that this is satisfied if condition (\ref{eq:sample_bound_disc}) holds.

%
%*********************************************************
%

\subsection{Proof of Theorem \ref{th:main_cont_time}}
\label{sec:proof_cont_theorem_upper_bound_appendix}

To prove Theorem \ref{th:main_cont_time} we recall that 
Proposition \ref{th:cond_to_hold} holds provided the appropriate
continuous time expressions are used for $\hG$ and $\hQ$, namely
\begin{align}
\hG & = - \nabla \cL (\Param^0_r) =  \frac{1}{T}\int_0^T\! x(t)\, \de b_r(t)\, , \nonumber \\
\hQ & =  \nabla^2 \cL (\Param^0_r) =   \frac{1}{T}\int_0^T\! x(t)x(t)^*
\, \de t \, .
\end{align}
These are of course random variables. In order to distinguish these from 
the discrete time version, we will adopt the notation $\hG^n$, $\hQ^n$
for the latter. We claim that these random variables can be coupled
(i.e. defined on the same probability space) in such a way that 
$\hG^n\to \hG$ and $\hQ^n\to\hQ$ almost surely as $n\to\infty$
for fixed $T$. Under assumption (\ref{eq:sample_bound_cont}),
and making use of Lemma \ref{th:equiv_of_bound_in_the_limit}
it is easy to show that (\ref{eq:sample_bound_disc})
holds for all $n>n_0$ with $n_0$ a sufficiently large constant.

Therefore, by the proof of Theorem \ref{th:main_discrete},
the conditions in Proposition \ref{th:cond_to_hold} hold 
for gradient $\hG^n$ and Hessian $\hQ^n$ for any $n\ge n_0$, with probability 
larger than $1-\delta$. But by the claimed convergence 
$\hG^n\to \hG$ and $\hQ^n\to\hQ$, they hold also for $\hG$ and $\hQ$
with probability at least $1-\delta$ which proves the theorem.

We are left with the task of showing that the discrete and continuous
time processes can be coupled in such a way that $\hG^n\to \hG$ 
and $\hQ^n\to\hQ$.
With slight abuse of notation, the state of the discrete time system 
\eqref{eq:DiscreteTimeModel} will be denoted by $x(i)$ where 
$i \in \naturals$ and the state of continuous time system 
\eqref{eq:BasicModel} by $x(t)$ where $t \in \reals$. We denote by $Q^0$ the solution of \eqref{eq:Lyapunov} and by $Q^0(\eta)$ the solution of \eqref{eq:Lyapunov2}. It is easy to check that $Q^0(\eta) \rightarrow Q^0$ as $\eta \rightarrow 0$ by the uniqueness of stationary state distribution. We couple the
process as follows.

The initial state of the continuous time system $x(t=0)$ is a 
$\normal(0, Q^0)$ random variable independent of $b(t)$ and the initial 
state of the discrete time system is defined to be 
$x(i=0) = (Q^0(\eta))^{1/2} (Q^0)^{-1/2} x(t = 0)$.
At subsequent times, $x(i)$ and $x(t)$ are assumed to be  generated by 
the respective dynamical systems using the same matrix $\Param^0$
using common randomness provided by the standard  Brownian motion 
$\{b(t)\}_{0\le t\le T}$ in $\reals^p$.
In order to couple $x(t)$ and 
$x(i)$, we construct $w(i)$, the noise 
driving the discrete time system, by letting 
$w(i)\equiv (b(Ti/n) - b(T(i-1)/n))$. 

The almost sure convergence $\hG^n\to \hG$ and $\hQ^n\to\hQ$
follows then from standard convergence of random walk to Brownian motion.

\begin{lemma}
\label{th:equiv_of_bound_in_the_limit}
Let $\sigma_{\max} \equiv \sigma_{\max} (I + \eta \Param^0)$ and $\rho_{\min} = - \lambda_{\max}(( \Param^0 + (\Param^0)^*) /2) > 0$ then,
\begin{align}
-\lambda_{\min}\left(\frac{\Param^0 + (\Param^0)^*}{2}\right) &\geq \limsup_{\eta \rightarrow 0} \frac{1 - \sigma_{\max}}{\eta},\\
\liminf_{\eta \rightarrow 0} \frac{1 - \sigma_{\max}}{\eta} &\geq -\lambda_{\max}\left(\frac{\Param^0 + (\Param^0)^*}{2}\right).
 \end{align} \label{th:limit_D}
\end{lemma}
\begin{proof}
 \begin{align}
  \frac{1 - \sigma_{\max}}{\eta} &= \frac{1 -  \lambda^{1/2}_{\max}((I + \eta \Param^0)^* (I + \eta \Param^0))}{\eta} \\
  &= \frac{1 -  \lambda^{1/2}_{max}(I + \eta (\Param^0 + (\Param^0)^*) + \eta^2 (\Param^0)^* \Param^0 ) }{\eta} \\
  &= \frac{1 - (1 + \eta u^* (\Param^0 + (\Param^0)^* + \eta (\Param^0)^* \Param^0 ) u)^{1/2}}{\eta},
 \end{align}
where $u$ is some unit vector that depends on $\eta$. Thus, since $\sqrt{1+x} = 1 + x/2 + \bigo(x^2)$,
\begin{align}
\liminf_{\eta \rightarrow 0} \frac{1 - \sigma_{\max}}{\eta} & = - \limsup_{\eta \rightarrow 0} u^*\left( \frac{\Param^0 + (\Param^0)^*}{2} \right) u \nonumber \\
& \geq -\lambda_{\max}\left(\frac{\Param^0 + (\Param^0)^*}{2} \right).
\end{align}
The other inequality is proved in a similar way.
\end{proof}

%
%****************************************************
%

\subsection{Proof of Theorem \ref{th:cont_reg_graph_bounded_degree}}
\label{sec:proof_of_reg_graph_upper_bound}

In order to prove Theorem \ref{th:cont_reg_graph_bounded_degree} we first state and prove the following lemma,
\begin{lemma}
Let $G$ be a simple connected graph of vertex degree bounded above by $k$. Let $\tilde{\Param}$ be its adjacency matrix and $\Param^0 = -h I + \tilde{\Param}$ with $h > k$. Then, for this $\Param^0$, the system in \eqref{eq:BasicModel} has $Q^0 = -(1/2) (\Param^0)^{-1}$ and,
\begin{align}
 \III Q^0_{(S^0)^C, {S^0}} (Q^0_{{S^0},{S^0}})^{-1}  \III _{\infty} & =  \III (\Param^0_{(S^0)^C, (S^0)^C})^{-1} \Param^0_{(S^0)^C, S^0} \III _{\infty} \nonumber \\
 \qquad & \leq k/h.
\end{align}
\label{th:cont_time_incoher_bound}
\end{lemma}
\begin{proof}
$\tilde{\Param}$ is symmetric so $\Param^0$ is symmetric. Since $\tilde{\Param}$ is irreducible and non-negative, Perron-Frobenious theorem tells that $\lambda_{\max}(\tilde{\Param}) \leq k$ and consequently $\lambda_{\max}(\Param^0) \leq -h + \lambda_{\max}(\tilde{\Param}) \leq -h + k$. Thus $h > k$ implies that $\Param^0$ is negative definite and using equation \eqref{eq:Lyapunov} we can compute  $Q^0 = -(1/2) (\Param^0)^{-1}$.

Now notice that, by the block matrix inverse formula, we have
\begin{align}
(Q^0_{{S^0},{S^0}})^{-1} &=  -2 C^{-1},\\
Q^0_{(S^0)^C, S^0} &= \frac{1}{2}( (\Param^0_{(S^0)^C, (S^0)^C})^{-1} \Param^0_{(S^0)^C, {S^0}} C),
\end{align}
where $C = \Param^0_{{S^0},{S^0}} - \Param^0_{{S^0}, (S^0)^C} (\Param^0_{(S^0)^C, (S^0)^C})^{-1} \Param^0_{(S^0)^C, {S^0}}$ and thus
\begin{equation}
 \III Q^0_{(S^0)^C, {S^0}} (Q^0_{{S^0},{S^0}})^{-1}  \III _{\infty} =  \III (\Param^0_{(S^0)^C, (S^0)^C})^{-1} \Param^0_{(S^0)^C, {S^0}} \III _{\infty}.
\end{equation}
Recall the definition of the infinity norm of a matrix $B$, $ \III B \III _\infty$,
\begin{equation}
 \III B \III _\infty = \max_i \sum_j |B_{ij}|.\label{eq:infty_mat_norm}
\end{equation}
Let $z = h^{-1}$ and write,
\begin{align}
(\Param^0_{(S^0)^C, (S^0)^C})^{-1} & = -z (I - z \tilde{\Param}_{(S^0)^C, (S^0)^C})^{-1}\nonumber \\
& =  -z \sum^\infty_{n = 0} (z \tilde{\Param}_{(S^0)^C, (S^0)^C})^n,\\
\Param^0_{(S^0)^C, S^0} &= z^{-1} z \tilde{\Param}_{(S^0)^C, {S^0}}.
\end{align}
This allows us to conclude that $$ \III (\Param^0_{(S^0)^C, (S^0)^C})^{-1} \Param^0_{(S^0)^C, {S^0}} \III _{\infty}$$ is in fact the maximum over all path generating functions of paths starting from a node $i \in (S^0)^C$ and hitting  ${S^0}$ for a first time. Let $\Omega_i$ denote this set of paths, $\omega$ a general path in $G$ and $|\omega|$ its length. Let $k_1, ..., k_{|\omega|}$ denote the degree of each vertex visited by $\omega$ and note that $k_m \leq k, \forall m$. Then each of these path generating functions can be written in the following form,
\begin{equation}
\sum_{\omega \in \Omega_i } z^{|\omega|}  \leq \sum_{\omega \in \Omega_i } \frac{1}{k_1...k_{|\omega|}}(kz)^{|\omega|} = \E_G ((kz)^{T_{i,{S^0}}}),\label{eq:inc_cond_path}
\end{equation}
where $T_{i,{S^0}}$ is the first hitting time of the set ${S^0}$ by a random walk that starts at node $i \in {S^0}^C$ and moves with equal probability to each neighboring node. But $T_{i,{S^0}} \geq 1$ and $kz < 1$ so the previous expression is upper bounded by $kz$.
\end{proof}

Now what remains to complete the proof of Theorem \ref{th:cont_reg_graph_bounded_degree} is to compute
values for the upper bound constants
$\alpha$, $\param_{\min}$, $\rho_{\min}$ and $C_{\min}$ in Theorem \ref{th:main_cont_time} . From Lemma \ref{th:cont_time_incoher_bound} we can set $\alpha = 1 - k/(k+m)$. In addition, clearly, we can choose $\param_{\min} = 1$. We also have that $\sigma_{\min} (\Param^0) \geq k+m - \sigma_{\max} (\tilde{\Param}) \geq m+k - k = m$ so we set $\rho_{\min} = m$. Finally, notice that 
\begin{align}
\lambda_{\min}(Q^0_{{S^0},{S^0}}) & = \frac{1}{2}\lambda_{\min}(-(\Param^0)^{-1}) \nonumber \\
& = \frac{1}{2} \frac{1}{\lambda_{\max}(-\Param^0)} \nonumber \\
& \geq \frac{1}{2} \frac{1}{m+k+k} \nonumber \\
& \geq \frac{1}{4(m+k)},
\end{align}
where in the last step we made use of the fact that $m+k > k$. 
Hence, we choose $C_{\min} = 1/(4(m + k))$.
Substituting these values in the inequality from Theorem  \ref{th:main_cont_time} gives the desired result.

%
%**************************
%

\subsection{Proofs of auxiliary results for the discrete-time model}
\label{sc:auxi_upper_bound}

\begin{figure*}[t]
\begin{equation}
\tilde{R}(j) =  \left(
\begin{array}{ccccccccc}
0&0& \ldots&0 &0 & 0 & \ldots &0 &0\\
\vdots&\vdots& \ddots&\vdots &\vdots & \vdots & \ddots & \vdots& \vdots\\
0&0& \ldots&0 &0 & 0 &\ldots & 0&0\\
\rho(m,j)&\rho(m-1,j)&\ldots&\rho(1,j)&\rho(0,j)&0&\ldots&0&0\\
\rho(m+1,j)&\rho(m,j)&\ldots&\rho(2,j)&\rho(1,j)&\rho(0,j)&\ldots&0&0\\
\vdots&\vdots&\ddots&\vdots&\vdots&\vdots&\ddots&0&0\\
\rho(m+n-1,j)&\rho(m+n-2,j)&\ldots&\rho(n,j)&\rho(n-1,j)&\rho(n-2,j)&\ldots&\rho(0,j)&0\\
\end{array}
\right).
\label{eq:rtilde_def}
\end{equation}
\hrulefill
\vspace{4pt}
\end{figure*}

\subsubsection{Proof of Proposition \ref{th:cond_to_hold}}
\label{sec:proof_of_norm_bounds_to_prove_for_cont_theorem}
In order to prove Proposition \ref{th:cond_to_hold} we first introduce two technical lemmas.

The following Lemma is taken from the proof of Lemma 6 in  \cite{ravikumar2008high}.

\begin{lemma}
For any subset $S \subseteq [p]$ the following decomposition holds,
\begin{equation}
\hQ_{S^C, S} \left(\hQ_{S,S}\right)^{-1} = T_1 + T_2 + T_3 + Q^0_{S^C, S} \left(Q^0_{S,S}\right)^{-1},
\end{equation}
where,
\begin{align}
T_1 &= Q^0_{S^C, S} \left( \left(\hQ_{S,S}\right)^{-1} - \left(Q^0_{S,S}\right)^{-1}\right), \\
T_2 &= (\hQ_{S^C, S} - Q^0_{S^C, S} ) \left(Q^0_{S,S}\right)^{-1} \text{   \;\;\; and \;\;\;    }\\
T_3 &= (\hQ_{S^C, S} - Q^0_{S^C, S})\left (\left(\hQ_{S,S}\right)^{-1} - \left(Q^0_{S,S}\right)^{-1}\right).\\
\end{align}
In addition, if $|S| \leq k$, if $ \III Q^0_{S^C, S} \left(Q^0_{S,S}\right)^{-1} \III _\infty < 1$ and $\lambda_{\min}(\hQ_{S,S}) \geq C_{\min}/2>0$ the following relations hold,
\begin{align}
 \III T_1 \III _\infty &\leq \frac{2 \sqrt{k} }{C_{\min}} \III \hQ_{S, S} - Q^0_{S, S} \III _\infty,\\
 \III T_2 \III _\infty &\leq \frac{\sqrt{k} }{C_{\min}} \III \hQ_{S^C, S} - Q^0_{S^C, S} \III _\infty \text{   \;\;\; and  \;\;\;  }\\
 \III T_3 \III _\infty &\leq \frac{2 \sqrt{k} }{C_{\min}^2} \III \hQ_{S^C, S} - Q^0_{S^C, S} \III _\infty \III \hQ_{S, S} - Q^0_{S, S} \III _\infty.
\end{align} \label{th:inco_decomp}
\end{lemma}

The following lemma, directly obtained from the proofs of Proposition 1 in \cite{ravikumar2008high} and Proposition 1 in \cite{zhao} respectively,
resumes the conditions that guarantee correct signed-support reconstruction
of $\Param_r^0$.

\begin{lemma}
If $\hQ_{{S^0},{S^0}} > 0$, then the dual vector $\hat{z}$ from the KKT conditions of the optimization problem \eqref{eq:discr_reg_prob} satisfies the following inequality,
\begin{align}
\|\hat{z}_{(S^0)^C}\|_\infty  & \leq  \III  \hQ_{(S^0)^C, {S^0}} \left(\hQ_{{S^0}, {S^0}}\right)^{-1}  \III _\infty \nonumber \\
& \qquad \left( 1 +  \frac{\|\hG_{S^0}\|_\infty}{\lambda}  \right) + \frac{\|\hG_{(S^0)^C}\|_\infty}{\lambda} \label{eq:bound_on_dual_vector}.
\end{align}
In addition, if 
\begin{equation}\label{eq:important_condition_on_gradient}
\|\hG_{S^0}\|_\infty \leq \frac{\param_{\min} \lambda_{\min}(\hQ_{{S^0},{S^0}}) }{2 k} - \lambda
\end{equation}
then $\|\Param^0_r - \hat{\Param}_r\|_\infty \leq \param_{\min}/2$. The same result holds for problem \eqref{eq:cont_reg_prob}.
\label{th:main_cond}
\end{lemma}

\begin{proofof}{Proposition \ref{th:cond_to_hold}}\label{pf:compact_condition}
To guarantee that our estimated support is at least contained in the true support we need to impose that $\|\hat{z}_{S^C}\|_\infty < 1$. To guarantee that we do not introduce extra elements in estimating the support and also to determine the correct sign of the solution we need to impose that $\|\Param^0_r - \hat{\Param}_r\|_\infty \leq \param_{\min}/2$.

Now notice that since $\lambda_{\min}(Q^0_{{S^0},{S^0}}) = C_{\min}$ the relation $\lambda_{\min}(\hQ_{{S^0},{S^0}}) \geq C_{\min}/2$ is guaranteed as long as $ \III \hQ_{{S^0}, {S^0}} - Q^0_{{S^0}, {S^0}} \III _\infty \leq C_{\min}/2$. Using the norm triangle inequality together with Lemma \ref{th:inco_decomp} to bound $\III...\III_{\infty}$ in \eqref{eq:bound_on_dual_vector}, it is easy to see that the bounds of Proposition \ref{th:cond_to_hold} lead to $|\hat{z}_{(S^0)^C}|_{\infty} < 1$ and to \eqref{eq:important_condition_on_gradient} being verified. In turn, these lead to a correct recovery of the sign and support of $\Param^0_r$.   
\end{proofof}

%
%*********************************************************
%

\subsubsection{Proof of Propositions \ref{th:gradie_prob_bound} and \ref{en:bddQij}}
\label{sec:proof_of_concentration_bounds}
\begin{figure*}
\begin{equation}
\Phi_j = \left(
\begin{array}{ccccccccc}
\rho(m,j)&\rho(m-1,j)&\ldots&\rho(1,j)&\rho(0,j)&0&\ldots&0\\
\rho(m+1,j)&\rho(m,j)&\ldots&\rho(2,j)&\rho(1,j)&\rho(0,j)&\ldots&0\\
\vdots&\vdots&\ddots&\vdots&\vdots&\vdots&\ddots&0\\
\rho(m+n-1,j)&\rho(m+n-2,j)&\ldots&\rho(n,j)&\rho(n-1,j)&\rho(n-2,j)&\ldots&\rho(0,j)\\
\end{array}
\right).
\label{eq:phi_j_def}
\end{equation}
\hrulefill
\vspace{4pt}
\end{figure*}
To prove the concentration bounds of Propositions \ref{th:gradie_prob_bound} and \ref{en:bddQij} we need the following lemmas.

\begin{lemma}\label{proof:first_bound_egein_values}
Let $r,j \in [p]$ and let $\rho(\tau,j)$ represent a $p \times p$ matrix with all rows equal to zero except the $r^{th}$ row which equals the $j^{th}$ row of ${(I + \eta \Param^0)}^\tau$ (the $\tau^{th}$ power of $I + \eta \Param^0$ ). Let $\tilde{R}(j) \in \reals^{(n+m+1)p\times(n+m+1)p}$ be defined as in Eq.~\ref{eq:rtilde_def}.

Define $R(j) = 1/2 (\tilde{R} + \tilde{R}^*)$ and let $\nu_i$ denote its $i^{th}$ eigenvalue and assume $\sigma_{\max} \equiv \sigma_{\max} (I + \eta \Param^0) < 1$. Then,
\begin{align}
\sum^{p(n+m+1)}_{i=1}\nu_i &= 0,\label{eq:gradient_sum_eigen_bound}\\
\max_i |\nu_i|  &\leq \frac{1}{1-\sigma_{\max}} \, \text{   \;\;\; and \;\;\;   } \label{eq:gradient_eigen_bound}\\
\sum^{p(n+m+1)}_{i=1}\nu^2_i &\leq \frac{1}{2} \frac{n}{1-\sigma_{\max}}.\label{eq:gradient_sum_eigen2_bound}
\end{align} \label{th:matrix_proper_grad}
\end{lemma}
\begin{proof}
First it is immediate to see that $\sum^{p(n+m+1)}_{i=1}\nu_i = Tr(R) = 0$. Let ${I_1}_\tau$ represent a $(n+m+1) \times (n+m+1)$ matrix with zeros everywhere and ones in the block-positions of $R(j)$ where $\rho(\tau,j)$ appears and ${I_2}_\tau$ represent a similar matrix but with ones in the block-position of $R(j)$ where $\rho(\tau,j)^*$ appears. Then $R$ can be written as,
\begin{equation}
R = \frac{1}{2} \left(  \sum^{m+n-1}_{\tau =  0} {I_1}_\tau \otimes   \rho(\tau,j)   +  {I_2}_\tau \otimes   \rho(\tau,j)^*  \right),
\end{equation}
where $\otimes$ denotes the Kronecker product of matrices. 
This expression can be used to compute an upper bound on $|\nu_i|$. Namely,
\begin{align}
\max_i{|\nu_i|} &= \sigma_{\max}(R) \nonumber \\
& \leq  \sum^\infty_{\tau = 0} \sigma_{\max} ( {I_1}_\tau \otimes   \rho(\tau,j) ) \nonumber \\
& \leq  \sum^\infty_{\tau = 0} \sigma_{\max} ( {I_1}_\tau)\sigma_{\max} (  \rho(\tau,j)  ) \nonumber \\
& \leq   \sum^\infty_{\tau = 0} \sigma_{\max} ( \rho(\tau,j) ) \nonumber \\
& \leq  \sum^\infty_{\tau = 0} \sigma_{\max}^{\tau} =\frac{1}{1-\sigma_{\max}}.
\end{align}
For the other bound we do,
\begin{align}
\sum^{(n+m+1)p}_{i=1}\nu^2_i &= Tr(R^2) \nonumber \\
& \leq  \frac{2}{4} \; n \; \sum^{\infty}_{\tau = 0} Tr(  \rho(\tau,j) \rho(\tau,j)^*  )\nonumber \\
&=\frac{1}{2}  n  \sum^{\infty}_{\tau = 0} \|  \rho(\tau,j)  \|^2_2 \nonumber \\
& \leq \frac{1}{2} n  \sum^{\infty}_{\tau = 0} \sigma_{\max}^{2\tau} \nonumber \\
& \leq \frac{1}{2} \frac{n}{1-\sigma_{\max}},
\end{align}
where in the last step we used the fact that $0 \leq \sigma_{\max} < 1$.
\end{proof}

\begin{lemma}
  Let $j \in [p]$. Define \footnote{Note that, with regards to Lemma \ref{th:matrix_proper_grad}, we are redefining the meaning of $\rho(\tau,j)$} $\rho(\tau,j) \in \reals^{1 \times p}$ to be the $j^{th}$ row of $(I + \eta \Param^0)^\tau$. Let $\Phi_j \in \reals^{n \times (n+m)p}$ be defined as in Eq.~\eqref{eq:phi_j_def}

Let $\nu_l$ denote the $l^{th}$ eigenvalue of the matrix $R(i,j) = 1/2 (\Phi_j^* \Phi_i + \Phi_i^* \Phi_j) \in \reals^{(n+m)p \times (n+m)p} $ (where $i \in [p]$) and assume $\sigma_{\max} \equiv \sigma_{\max} (I + \eta \Param^0) < 1$ then,
\begin{align}
|\nu_l| &\leq \frac{1}{(1-\sigma_{\max})^2} \text{   \;\;\; and  \;\;\;  }\\
\frac{1}{n}\sum^{(n+m)p}_{l = 1} \nu^2_l &\leq \frac{2}{(1-\sigma_{\max})^3}\left(1 + \frac{3}{2n} \frac{1}{1-\sigma_{\max}} \right).
\end{align} \label{th:matrix_proper_hess}
\end{lemma}
\begin{proof}
The first bound can be proved in a trivial manner. In fact, since for any matrix $A$ and $B$ we have $\sigma_{\max}(A+B) \leq \sigma_{\max}(A) + \sigma_{\max}(B)$ and $\sigma_{\max}(AB) \leq \sigma_{\max}(A)\sigma_{\max}(B)$ we can write
\begin{align}
\max_l |\nu_l| &= \sigma_{\max} (1/2 (\Phi_j^* \Phi_i + \Phi_i^* \Phi_j)) \nonumber \\
 & \leq 1/2 (\sigma_{\max}(\Phi_j^* \Phi_i) +  \sigma_{\max}(\Phi_i^* \Phi_j))\\
&\leq \sigma_{\max}(\Phi_i^* \Phi_j) \nonumber \\
& \leq \frac{1}{(1-\sigma_{\max})^2},
\end{align}
where in the last inequality we used the fact $\sigma_{\max}(\Phi_j) \leq 1/ (1-\sigma_{\max})$. The proof of this last fact is just a copy of the proof of the bound \eqref{eq:gradient_eigen_bound} in Lemma \ref{th:matrix_proper_grad}.

Now notice that $\Phi_i^* \Phi_j$ can be written as a block matrix
\begin{equation}
\left(
\begin{array}{cccc}
\tilde{A} &  \tilde{D}\\
\tilde{C} &  \tilde{B}
\end{array}
\right)
\end{equation}
where $\tilde{A}, \tilde{B}, \tilde{C}$ and $\tilde{D}$ are block-matrices. Each block is a $p$ by $p$ matrix. $\tilde{A}$ has $p \times p$ blocks, $\tilde{B}$ has $n \times n$ blocks, $\tilde{C}$ has $n \times m$ blocks and $\tilde{D}$ has $m \times n$ blocks. If we index the blocks of each matrix with the indices $x,y$ these can be described in the following way
\begin{align}
\tilde{A}_{xy} &=\sum^m_{s=1} \rho(m-x+s,i)^* \rho(m-y+s,j) \\
\tilde{B}_{xy} &=\sum^{n-x}_{s=0} \rho(s,i)^* \rho(s+x-y,j),\;\;x \geq y\\
\tilde{B}_{xy} &=\sum^{n-y}_{s=0} \rho(s+y-x,i)^* \rho(s,j),\;\;x \leq y\\
\tilde{C}_{xy} &=\sum^{n-x}_{s=0} \rho(s,i)^* \rho(m-y+x+s,j)\\
\tilde{D}_{xy} &=\sum^{n-y}_{s=0} \rho(m-x+y+s,i)^* \rho(s,j).
\end{align}
With this in mind and denoting by $A,B,C$ and $D$ the `symmetrized' versions of these same matrices (e.g.: $A = 1/2 (\tilde{A} + \tilde{A}^*)$) we can write,
\begin{equation}
\sum^{(n+m)p}_{l=1} \nu^2_l = Tr(R(i,j)^2) = Tr(A^2) + Tr(B^2) + 2 Tr(CD). 
\end{equation}
We now compute a bound for each one of the terms. We exemplify in detail the calculation of the first bound only.
First write,
\begin{align}
Tr(A^2) &= \sum^m_{x=1} \sum^m_{y=1} Tr(A_{xy} A^*_{xy}).
\end{align}
Now notice that each $Tr(A_{xy} A^*_{xy})$ is a sum over $\tau_1,\tau_2  \in [m]$ of terms of the type,
\begin{align}
& (\rho(m-x+\tau_1,i)^* \rho(m-y+\tau_1,j) \nonumber \\
& \qquad + \rho(m-x+\tau_1,j)^* \rho(m-y+\tau_1,i)) \nonumber \\
& \times(\rho(m-y+\tau_2,j)^* \rho(m-x+\tau_2,i)  \nonumber \\
& \qquad + \rho(m-y+\tau_2,i)^* \rho(m-x+\tau_2,j)).
\end{align}
The trace of a matrix of this type can be easily upper bounded by
\begin{align}
& (\sigma_{\max})^{m - x + \tau_1 + m - y + \tau_1 + m - y + \tau_2 + m - x + \tau_2} \nonumber \\
& \qquad = (\sigma_{\max})^{2(m-x) +2(m-y) +2 \tau_1 + 2\tau_2}
\end{align}
which finally leads to
\begin{equation}
Tr(A^2) \leq \frac{1}{(1-\sigma_{\max})^4}.
\end{equation}
Similarly for the other terms
\begin{align}
Tr(B^2) &\leq \sum^{n,n}_{x,y} \sum_{\tau_1,\tau_2} \sigma_{\max}^{2 \tau_1 + 2 \tau_2+ 2|x - y|} \leq \frac{2n}{(1 - \sigma_{\max})^3}\\
Tr(DC) &= \sum^m_{x=1}  \sum^n_{y=1} Tr(C_{xy} D_{yx}) \nonumber \\
& \leq \sum^{m,n,n-y,n-y}_{x,y,\tau_1,\tau_2} \sigma_{\max}^{2 (m-x) + 2y + 2\tau_1 + 2\tau_2} \nonumber \\
& \leq \frac{1}{(1 - \sigma_{\max})^4}.
\end{align}
Putting all these together leads to the desired bound.
\end{proof}

\vspace{0.5cm}

\begin{proofof}{Proposition \ref{th:gradie_prob_bound}}
We will start by proving that this exact same bound holds when the probability of the event $\{ \|\hG_S\|_\infty > \epsilon\}$ is computed with respect to a trajectory $\{x(t)\}^n_{t=-m}$ that is initiated at instant $t = -m$ with the value $x(-m) = w(-m)$. Assume we have done so. Now notice that as $m \rightarrow \infty$, $X$ converges in distribution to $n$ consecutive samples from the model \eqref{eq:DiscreteTimeModel} when this is initiated from stationary state. Since $\|\hG_S\|_\infty$ is a continuous function of $X = [x(0),...,x(n-1)]$, by the Continuous Mapping Theorem, $\|\hG_S\|_\infty$ converges in distribution to the corresponding random variable in the case when the trajectory $\{x(i)\}^n_{i=0}$ is initiated from stationary state. Since the probability bound does not depend on $m$ we have that this same bound holds for stationary trajectories too.

We now prove our initial claim. Recall that  $\hG_j = (X_j W_r^*)/(n\eta)$. Since $X$ is a linear function of
the independent Gaussian random variables $W$ we can write
$X_j W_r^* = \eta  z^* R(j) z$, where $z\in \reals^{p (n+m+1)}$ is a 
vector of i.i.d. $\normal(0,1)$ random variables and $R(j) \in \reals^{p(n+m+1)\times
p(n+m+1)}$ is the symmetric matrix defined in Lemma \ref{th:matrix_proper_grad}.

Now apply the standard Bernstein method. First by union bound we have
\begin{align*}
\prob\big\{\|\hG_S\|_\infty > \epsilon\big\} &\leq  2|S| 
\,\max_{j\in S} \prob\big\{z^* R(j) z > n \epsilon\big\} \, .
\end{align*}
Next denoting by $\{\nu_i\}_{1\le i\le p(n+m+1)}$ the eigenvalues
of $R(j)$, we have, for any $\gamma>0$,
\begin{align*}
& \prob\big\{z^* R(j) z > n \epsilon\big\} \nonumber \\
& \qquad = \prob\Big\{\sum^{p(n+m+1)}_{i=1}\nu_i z^2_i > n \epsilon \Big\} \\
&\qquad \leq  e^{-n\gamma\eps}\, \prod_{i=1}^{p(n+m+1)}
\E\big\{e^{\gamma\nu_i z_i^2}\big\}\\
& \qquad =\exp\left(-n \Big(\gamma \epsilon +  \frac{1}{2n} \sum^{(n+m+1)p}_{i=1} \log(1-2 \nu_i \gamma)  \Big)\right)\, .
\end{align*}
Let $\gamma = \frac{1}{2} (1-\sigma_{\max}) \epsilon$. Using the bound obtained for $\max_i{| \nu_i |}$  in Eq.~\eqref{eq:gradient_eigen_bound} (Lemma \ref{th:matrix_proper_grad}) we have $|2 \nu_i \gamma| \leq  \epsilon$. Now notice that if $|x|<1/2$ then $\log (1-x) > -x-x^2$. Thus, if we assume $\epsilon < 1/2$ and given that $\sum^{(n+m+1)p}_{i=1}\nu_i = 0$ (see Eq.~\eqref{eq:gradient_sum_eigen_bound} in Lemma \ref{proof:first_bound_egein_values}) we can continue the chain of inequalities,
\begin{align}
&\prob(\|\hG_S\|_\infty > \epsilon) \nonumber \\
& \qquad \leq 2 |S| \max_{j} \exp \left(-n \Big(\gamma \epsilon - 2 \gamma^2 \frac{1}{n} \sum^{(n+m+1)p}_{i=1} \nu^2_i \Big)\right) \nonumber \\
& \qquad \leq 2 |S| \exp \left(- \frac{n}{2} (1-\sigma_{\max}) \epsilon^2 \Big(1 + \frac{1}{2} \frac{1-\sigma_{\max}}{1-\sigma_{\max}}\Big)\right) \nonumber \\
& \qquad \leq 2 |S| \exp \left(-\frac{n}{4}  (1-\sigma_{\max}) \epsilon^2\right).
\end{align}
where the second inequality is obtained using the bound in Eq.~\eqref{eq:gradient_sum_eigen2_bound} from Lemma \ref{proof:first_bound_egein_values}.
\end{proofof}

%\comment{
%\begin{proposition}
%Let $i,j\in \{ 1,...,p\}$,  $\sigma_{\max} \equiv \sigma_{max}(I + \eta \Param^0) < 1$, $T = \eta n > 3/D$ and $0 < \epsilon< 2/D$ where $D = (1-\sigma_{\max})/\eta$ then, 
%\begin{align}
%\prob(|\hQ_{ij} - Q^0_{ij})|  > \epsilon) \leq   2 e^{-\frac{n}{32 \eta^2}  (1-\sigma_{max})^3 \epsilon^2}.
%\end{align} \label{en:bddQij}
%\end{proposition}
%}

\vspace{0.5 cm}

\begin{proofof}{Proposition \ref{en:bddQij}}
The proof is very similar to that of proposition \ref{th:gradie_prob_bound}. We will first show that the bound
\begin{equation}
\prob(|\hQ_{ij} - \E(\hQ_{ij}) |  > \epsilon) \leq   2 e^{-\frac{n}{32 \eta^2}  (1-\sigma_{max})^3 \epsilon^2},
\end{equation}
holds in the case where the probability measure and expectation are taken with respect to trajectories $\{x(i)\}^n_{i=0}$ that started at time instant $t = -m$ with $x(-m) = w(-m)$. Assume we have done so. Now notice that as $m \rightarrow \infty$, $X$ converges in distribution to $n$ consecutive samples from the model \ref{eq:DiscreteTimeModel} when this is initiated from stationary state. In addition, as $m \rightarrow \infty$, we have from lemma \ref{th:EQhat_Q_start_bound} that $\E(\hQ_{ij}) \rightarrow Q^0_{ij}$. Since $\hQ_{ij}$ is a continuous function of $X = [x(0),...,x(n-1)]$, a simple application of the Continuous Mapping Theorem plus the fact that the upper bound is continuous in $\epsilon$ leads us to conclude that the bound also holds when the system is initiated from stationary state.

To prove our initial statement first recall the definition of $\hQ$ and notice that we can write,
\begin{equation} \label{eq:matrix_formula_for_Q_hat}
\hQ_{ij}  =   \frac{\eta}{n} z^* R(i,j) z,
\end{equation}
where $z \in \reals^{m+n}$ is a vector of i.i.d. $\normal(0,1)$ and $R(i,j) \in \reals^{(n+m) \times (n+m)} $ is defined as in lemma \ref{th:matrix_proper_hess}. Letting $\nu_l$ denote the $l^{th}$ eigenvalue of the symmetric matrix $R(i,j)$ we can further write,
\begin{equation}
\hQ_{ij} - \E (\hQ_{ij}) = \frac{\eta}{n} \sum^{(n+m)p}_{l = 1} \nu_l (z^2_l - 1).
\end{equation}
By Lemma \ref{th:matrix_proper_hess} we know that,
\begin{align}
|\nu_l| &\leq \frac{1}{(1-\sigma_{\max})^2} \text{  \;\;\; and \;\;\;  }\\
\frac{1}{n}\sum^{(n+m)p}_{l = 1} \nu^2_l &\leq \frac{2}{(1-\sigma_{\max})^3}\left(1 + \frac{3}{2n} \frac{1}{1-\sigma_{\max}} \right) \nonumber \\
& \leq \frac{3}{(1-\sigma_{\max})^3},
\end{align}
where we applied $T > 3/D$ in the last line.

Now we are done since applying Bernstein method, this time with $\gamma = 1/8 \, (1-\sigma_{\max})^3 \epsilon / \eta$, and making again use of the fact that $\log(1-x)> -x-x^2$ for $|x|<1/2$ we get,
\begin{align}
& \prob ( \hQ_{ij} - \E (\hQ_{ij}) > \epsilon ) \nonumber \\
& \qquad = \prob\Big( \sum^{(n+m)p}_{l = 1} \nu_l (z^2_l - 1) > \epsilon n / \eta \Big)  \nonumber \\
&\qquad \leq e^{-\frac{\gamma \epsilon n}{\eta}} e^{- \gamma \sum^{(n+m)p}_{l=1} \nu_l} e^{-1/2 \sum^{(m+n)p}_{l=1} \log (1-2 \gamma \nu_l)} \nonumber \\
&\qquad \leq e^{-\frac{\gamma \epsilon n}{\eta} - \gamma \sum^{(n+m)p}_{l=1} \nu_l + \gamma \sum^{(n+m)p}_{l=1} \nu_l +  2 \gamma^2 \sum^{(n+m)p}_{l=1} \nu^2_l} \nonumber \\
&\qquad \leq e^{-\frac{n}{32 \eta^2} (1-\sigma_{\max})^3 \epsilon^2}.
\end{align}
Above, in order to apply the bound on $\log(1-x)$, we require that $\epsilon < 2/D$.

An analogous reasoning leads us to,
\begin{equation}
\prob ( \hQ_{ij} - \E (\hQ_{ij}) < - \epsilon ) \leq e^{-\frac{n}{32 \eta^2} (1-\sigma_{\max})^3 \epsilon^2}
\end{equation}
and the results follows.

\end{proofof}

\begin{lemma}
As before, assume $\sigma_{\max} \equiv \sigma_{\max} (I + \eta \Param^0) < 1 $ and consider that model \eqref{eq:DiscreteTimeModel} was initiated at time $-m$ with
$x(-m) = w(-m)$ then
\begin{equation}
|\E(\hQ_{ij}) - Q^0_{ij}| \leq \frac{1}{n+m} \frac{\eta}{( 1 - \sigma_{\max})^2  }.
\end{equation} \label{th:EQhat_Q_start_bound}
\end{lemma}

\begin{proof}
Let $\rho = I + \eta \Param^0$. 
Taking the expectation of $\hat{Q}_{ij}$ in \eqref{eq:matrix_formula_for_Q_hat}, and
recalling that $z$ is a vector of i.i.d. standard
Gaussian variables, we can write,
\begin{equation}\label{eq:sum_formula_for_Q_hat}
\E(\hQ_{ij}) = \eta \sum^{n+m-1}_{l = 0} \frac{m+n-l}{n+m} (\rho^l {\rho^*}^l)_{ij}.
\end{equation}
We also have that
\begin{equation}
Q^0_{ij} = \eta \sum^{\infty}_{l = 0} ({\rho}^l {\rho^*}^l)_{ij}.
\end{equation}
This last expression can be proved, for example, by taking
$n \rightarrow \infty$ in \eqref{eq:sum_formula_for_Q_hat}.
Putting these two expressions together we obtain
\begin{align}
& Q^0_{ij} - \E(\hQ_{ij}) \nonumber \\
& \qquad = \eta \left(\sum^{\infty}_{l = m+n} ({\rho}^l {{\rho}^*}^l)_{ij} + \sum^{n+m-1}_{l = 1} \frac{l}{m+n} ({\rho}^l {{\rho}^*}^l)_{ij} \right).
\end{align}
Using the fact that for any matrix $A$ and $B$ $\max_{ij}(A_{ij}) \leq \sigma_{\max} (A) $, $\sigma_{\max}(AB) \leq \sigma_{\max}(A)\sigma_{\max}(B)$ and $\sigma_{\max}(A + B) \leq \sigma_{\max}(A) + \sigma_{\max}(B)$, and introducing the notation $\zeta = \rho^2 $, we can write,
\begin{align}
|\E(\hQ_{ij}) - Q^0_{ij}| &\leq \eta \left( \frac{\zeta^{n+m}}{1-\zeta}  + \frac{\zeta}{n+m} \sum^{m+n-2}_{l = 0} \zeta^ l\right) \nonumber \\
& = \frac{\eta (\zeta^2 + \zeta^{n+m} -2 \zeta^{m+n+1})}{(m+n)(1 - \zeta)^2} \nonumber \\
&\leq  \frac{\eta}{(m+n)( 1 - \sigma_{max})^2  }.
\end{align}
Above, we used the fact that for $\zeta \in [0,1]$ and $n \in \naturals$ we have $1-\zeta \geq 1-\sqrt{\zeta}$ and $\zeta^2 + \zeta^n - 2\zeta^{1+n} \leq 1$.
\end{proof}

%
%*********************************************************
%

\vspace{0.5cm}

\subsubsection{Proof of Theorem \ref{th:main_discrete} for discrete case system}
\label{sc:full_main_disc_theo_proof}

In order to prove Theorem \ref{th:main_discrete} we need to compute the probability that the conditions given by Proposition \ref{th:cond_to_hold} hold.

From the statement of the theorem we have that the two conditions, $\alpha, C_{\min} > 0$, of Proposition \ref{th:cond_to_hold} hold.

In order to make the first condition on $\hG$ imply the second condition on $\hG$ we assume that 
\begin{equation}
\frac{\lambda \alpha}{3} \leq \frac{\param_{\min} C_{\min}}{4 k} - \lambda 	
\end{equation}
which is guaranteed to hold if
\begin{equation}
\lambda \leq \param_{\min} C_{\min} / 8k. \label{eq:main_disc_lambda} 
\end{equation}

We also combine the two last conditions on $\hQ$ and obtain the sufficient condition
\begin{equation} \label{eq:equivalent_condition_of_norm_of_Q}
 \III \hQ_{[p], {S^0}} - Q^0_{[p], {S^0}} \III _\infty \leq \frac{\alpha}{12} \frac{C_{\min}}{\sqrt{k}}. 
\end{equation}
Note that $[p]=S^0 \cup (S^0)^c$.

We then impose that both the probability that condition \eqref{eq:equivalent_condition_of_norm_of_Q} on $\hQ$ fails and the probability that condition \eqref{eq:Conditions0} on $\hG$ fails are upper bounded by $\delta/2$. 
Using Proposition \ref{th:gradie_prob_bound}, we can guarantee that the condition on $\hG$ fails with probability smaller than $\delta/2$ if we set
\begin{equation}
\lambda^2 = 36 \alpha^{-2} (n\eta D)^{-1} \log (4p/ \delta).
\end{equation}
Since we also want \eqref{eq:main_disc_lambda} to be
satisfied, we substitute $\lambda$ from the previous expression in \eqref{eq:main_disc_lambda} and we conclude that $n$ must satisfy
\begin{equation}\label{eq:bdd_n1}
n \geq 2304 k^2 {C_{\min}}^{-2} {\param_{\min}}^{-2} \alpha^{-2} (D \eta)^{-1} \log(4p/ \delta).
\end{equation}
Since in addition, the application of the probability bound in Proposition \ref{th:gradie_prob_bound} requires that
\begin{equation}
\frac{\lambda^2 \alpha^2}{9} < 1/4,
\end{equation}
we need to impose further that,
\begin{equation}\label{eq:bdd_n2}
n \geq 16 (D \eta)^{-1} \log(4 p /\delta).
\end{equation}

To use Corollary \ref{th:prob_bound_inft_matrix_hess} for computing the probability that the condition on $\hQ$ holds we need that,
\begin{equation}\label{eq:bdd_n3}
n \eta > 3/D,
\end{equation}
and
\begin{equation}
\frac{\alpha C_{\min}}{12 \sqrt{k}} <  2k D^{-1}.
\end{equation}
The last expression imposes the following conditions on $k$,
\begin{equation}
k^{3/2} > 24^{-1}  \alpha C_{\min} D \label{eq:main_th_restric_k}.
\end{equation}

We then have that the condition on $\hQ$ holds with probability smaller than $1/2$ if 
\begin{equation}\label{eq:bdd_n4}
n > 4608 \eta^{-1} k^{3} \alpha^{-2} {C_{\min}}^{-2} D^{-3} \log {4pk / \delta}.
\end{equation}
Note that the restriction \eqref{eq:main_th_restric_k} on $k$ looks unfortunate but, since $k \geq 1$, we can actually show it always holds. Just notice $\alpha < 1$ and that
\begin{align}
& \sigma_{\max}(Q^0_{{S^0},{S^0}}) \leq \sigma_{\max}(Q^0) \leq \frac{\eta}{1-\sigma_{\max}} \nonumber \\
& \qquad \Leftrightarrow D \leq \sigma_{\max}^{-1}(Q^0_{{S^0},{S^0}})
\end{align}
therefore $C_{\min} D \leq \sigma_{\min}(Q^0_{{S^0},{S^0}}) / \sigma_{\max}(Q^0_{{S^0},{S^0}}) \leq 1$. This last expression also allows us to simplify the four restrictions on $n$ (namely \eqref{eq:bdd_n1}, \eqref{eq:bdd_n2}, \eqref{eq:bdd_n3} and \eqref{eq:bdd_n4}) into a single one that dominates them. In fact, since $C_{\min} D \leq 1$ we also have $C_{\min}^{-2} D^{-2} \geq C_{\min}^{-1} D^{-1} \geq 1$ and this allows us to conclude that the only two conditions on $n$ that we actually need to impose are the one at Equations \eqref{eq:bdd_n1}, and \eqref{eq:bdd_n4}. A little more of algebra shows that these two inequalities are satisfied if

\begin{equation}
n \eta > \frac{10^4  k^2 (k D^{-2} + \param_{\min}^{-2})}{\alpha^{2} D C_{\min}^{2} }  \log(4p k/ \delta).
\end{equation}
This concludes the proof of Theorem \ref{th:main_discrete}.\\

%
%*********************************************************
%

\section{Proofs of the Lower Bounds on the Sample-Complexity of General Reconstruction Algorithms}

In this section we prove Theorem \ref{th:linear_lbound_sparse} and Theorem
\ref{th:main_lbound} to Theorem \ref{th:non_linear_lbound}.
 
Throughout, $\{x(t)\}_{t\ge 0}$ 
is assumed to be a stationary process. It is immediate to check
that under the assumptions of the Theorems \ref{th:linear_lbound_sparse} and \ref{th:linear_lbound_dense}, the SDE admit a unique stationary measure, with bounded covariance $Q^0$. Recall that
\begin{align}
Q^0 &= \E\{x(0) x(0)^*\} - \E\{x(0)\}(\E\{x(0)\})^*\\
& =\E\{x(t) x(t)^*\} - \E\{x(t)\}(\E\{x(t)\})^*.
\end{align}

%
%*********************************************************
%

\subsection{A general bound for linear SDEs}

Before passing to the actual proofs, it 
is useful to establish a general bound for linear SDEs  
(\ref{eq:BasicModelLin}) with symmetric
interaction matrix $\Param^0$. 
\begin{lemma} \label{th:mut_inf_bound_linear}
Assume that  $\{x(t)\}_{t\ge 0}$ is a stationary process
generated by the linear SDE (\ref{eq:BasicModelLin}), with $\Param^0$ symmetric.
Let $\hM_T(X^T)$ be an estimator of $M(\Param^0)$ based on $X^T$. 
If  $\prob(\hM_T(X^T) \neq M(\Param^0) ) <\frac{1}{2}$ then
\begin{equation} \label{eq:bound_linear}
T \geq \frac{H(M(\Param^0)) - \log(|\mathcal{M}|) - 2 I(\Param^0;x(0)) - 2}{ \frac{1}{2}\tr
\{ \E \{-\Param^0 \}-   (\E\ \{ -{(\Param^0)}^{-1} \})^{-1} \} \}},
\end{equation}
where $|\mathcal{M}|$ is the size of the alphabet
of $M(\Theta^0)$.
\end{lemma}
\begin{proof}
The bound follows from Corollary \ref{th:main_simpler_lbound}
after showing that
\begin{align}
& \E_{x(0)}  \{ \Var_{\Param^0|x(0) } (\Param^0 x(0))) \nonumber \\
& \qquad \le (1/2)\tr\{ \E \{-\Param^0 \}- (\E\ \{ -{(\Param^0)}^{-1} \})^{-1} \}.
\end{align}

First note that
\begin{align}
& \E_{x(0)}  \{ \Var_{\Param^0|x(0) } (\Param^0 x(0)) \}  \nonumber \\
& \qquad = \E_{x(0)} \|  \Param^0 x(0)  - \E_{\Param^0|x(0) } (\Param^0 x(0)|x(0)) \|^2_2.\label{eq:norm2_err}
\end{align}
The quantity in \eqref{eq:norm2_err} can be thought of as the $\ell_2$-norm error of estimating $\Param^0 x(0)$ based on $x(0)$ using $\E_{\Param^0|x(0) }(\Param^0 x(0)|x(0))$. Since conditional expectation is the minimal mean square error estimator, replacing 
$\E_{\Param^0|x(0) } (\Param^0 x(0)|x(0))$ by any estimator of $\Param^0 x(0)$ based on $x(0)$ gives an upper bound for the expression in \eqref{eq:norm2_err}. We choose as an estimator a linear estimator, i.e., an estimator of the form $B x(0)$ where $B = (\E_{\Param^0} \Param^0 Q^0) (\E_{\Param^0} Q^0)^{-1}$. We then have
\begin{align} \label{eq:err_linear_est}
 &\E_{x(0)} ||  \Param^0 x(0)  - \E_{\Param^0|x(0) } (\Param^0 x(0)|x(0)) ||^2_2 \nonumber \\
& \qquad \leq  \E_{x(0)} ||  \Param^0 x(0)  - B x(0) ||^2_2  \nonumber \\
& \qquad =\tr \{ \E \{\Param^0 x(0) (x(0))^* {\Param^0}^*\} \} \nonumber \\
& \qquad \qquad - 2 \tr \{  B \E \{ x(0) (x(0))^* {\Param^0}^*\} \} \nonumber\\
& \qquad \qquad + \text{Tr} \{ B \E\{ x(0) (x(0))^* \} B^* \}.
\end{align}
Furthermore, for a linear system, $Q^0$ satisfies the Lyapunov equation $\Param^0 Q^0  + Q^0 {(\Param^0)}^* + I = 0$.
For $\Param^0$ symmetric, this implies $Q^0 = -(1/2) {(\Param^0)}^{-1}$. 
Substituting this expression in \eqref{eq:norm2_err} and 
\eqref{eq:err_linear_est} finishes the proof.
\end{proof}

%
%*********************************************************
%

\subsection{Proof of Theorem \ref{th:linear_lbound_sparse}}

We prove Theorem \ref{th:linear_lbound_sparse} by showing that the same complexity bound holds in the case when we are trying to estimate the signed support of $\Param^0$ for an $\Param^0$ that is uniformly randomly chosen with a distribution supported on ${\cA}^{(S)}$ and we simultaneously require that the average probability of error is smaller than $1/2$.
This guarantees that, unless the bound holds, there exists $A \in 
{\cA}^{(S)}$ for which the probability of error is bigger than $1/2$.
The complexity bound for random matrices $\Param^0$ is proved using Lemma \ref{th:mut_inf_bound_linear} together with Lemma \ref{th:random_matrix_calc_sparse} about random matrices.

More specifically, we generate $\Param^0$ at random as follows.
Let $G$ be the random matrix constructed from 
the adjacency matrix of a uniformly random $k$-regular graph.
Generate $\tilde{\Param^0}$ 
by flipping the sign of each non-zero entry in  $G$ 
with probability $1/2$ independently. We  define 
$\Param^0$ to be the random matrix $\Param^0 = -(\gamma + 2\param_{\min}\sqrt{k-1})I + \param_{\min} \tParam^0$ where  $\gamma=\gamma(\tParam^0)>0$ is the smallest value such that the 
maximum eigenvalue of $\Param$ is smaller than $-\rho$.
This guarantees that $\Param^0$ satisfies the four properties
of the class $\cA^{(S)}$.

The following lemma encapsulates the necessary random matrix calculations to prove the complexity bound
for random matrices.
\begin{lemma} \label{th:random_matrix_calc_sparse}
Let $\Param$ be a random matrix defined as above and 
\begin{align*}
Q(\param_{\min},k,\rho)\equiv \lim_{p \rightarrow \infty} \frac{1}{p} \{ \tr \{ \E(-\Param) \}  - \tr \{ (\E(-\Param^{-1}))^{-1} \} \}.
\end{align*}
Then, there exists a constant  $C'$ only dependent on $k$ such that
\begin{align}
&Q(\param_{\min},k,\rho)  \leq \min \Big \{\frac{C' k \param^2_{\min}}{\rho} ,\frac{k \param_{\min}}{\sqrt{k-1}}\Big\}.
\end{align}
\end{lemma}
\begin{proof}
First notice that 
\begin{align} \label{eq:linear_bound_first_term}
\lim_{p \rightarrow \infty } \frac{1}{p} \E \text{Tr} \{-\Param\} & = \lim_{p \rightarrow \infty} \E(\gamma) + 2\param_{\min}\sqrt{k-1} \nonumber \\
&= \rho +  2\param_{\min}\sqrt{k-1}
\end{align}
since by Kesten-McKay law \cite{friedman}, for large $p$, the spectrum of $\tParam$ has support in $(- \epsilon -2\param_{\min}\sqrt{k-1},2\param_{\min}\sqrt{k-1} + \epsilon)$ with high probability. Notice that unless we randomize each entry of $\tParam$ with $\{-1,+1\}$ values, every $\tParam$ will have $k$ as its largest eigenvalue and the above limit will not hold.

For the second term, $\text{Tr}\{ (\E(-\Param^{-1}))^{-1} \}$,  we will compute a lower bound. For that purpose let $\lambda_i > 0$ be the $i^{th}$ eigenvalue of the matrix $\E(-\Param^{-1})$. We can write,
\begin{align} \label{eq:jensen_trick}
\frac{1}{p} \text{Tr} \{ (\E(-\Param^{-1}))^{-1} \} & = \frac{1}{p} \sum^p_{i = 1} \frac{1}{\lambda_i} \nonumber \\
& \geq \frac{1}{\frac{1}{p} \sum^p_{i = 1} \lambda_i} \nonumber \\
& =\frac{1}{ \E  \{ \frac{1}{p} \text{Tr} \{ (-\Param)^{-1}  \}\} }
\end{align}
where we applied Jensen's inequality in the last step. By Kesten-McKay law we now have that,
\begin{align}\label{eq:linear_bound_second_term}
&\lim_{p \rightarrow \infty} \E  \Big\{ \frac{1}{p} \text{Tr} \{ (-\Param)^{-1}  \}\Big\}  = \E \Big\{  \lim_{p \rightarrow \infty}    \frac{1}{p} \text{Tr} \{ (-\Param)^{-1}  \}   \Big\} \nonumber \\
& = \frac{1}{\param_{\min}} G(k,\rho/\param_{\min}+2\sqrt{k-1})
\end{align}
where
\begin{align}\label{eq:integral_def_G}
G(k,z) = \int \frac{-1}{\nu - z} \de \mu(\nu). 
\end{align}
Above, $\mu(\nu)$ is the Kesten-McKay distribution and, inside its support, $\nu \in [-2\sqrt{k-1},-2\sqrt{k-1}]$, it is defined by
\begin{align*} \label{eq:mckay_law}
\de \mu(\nu)  = \frac{k}{2 \pi} \frac{\sqrt{4(k-1) - \nu^2 }}{k^2 - \nu^2} \de \nu.
\end{align*}
The integral \eqref{eq:integral_def_G} can be computed exactly 
\begin{align}
G(k,z)  = -\frac{(k-2) z-k \sqrt{-4 k+z^2+4}}{2 \left(z^2-k^2\right)}. 
\end{align}
From the closed form expression for $G(k,z)$ one can
see that
\begin{align}
 &\lim_{\rho \rightarrow 0} Q(\param_{\min},k,\rho) = \frac{\param_{\min} k}{\sqrt{k-1}} \text{  \;\;\; and \;\;\;  } \\
 &\lim_{\rho \rightarrow \infty}  \rho \, Q(\param_{\min},k,\rho)  = k (\param_{\min})^2. 
\end{align}
Finally, notice that $Q(\param_{\min},k,\rho)/\param_{\min}$ can be see as function of $k$ and $\rho/\param_{\min}$ alone. In addition, because it is strictly decreasing with $\rho/\param_{\min}$, the limits above
imply that $Q(\param_{\min},k,\rho)/\param_{\min} \leq k / \sqrt{k-1}$ and that there is a large $C'$ such that $Q(\param_{\min},k,\rho)/\param_{\min} \leq C' k \theta_{\min} / \rho$ for $\rho$ sufficiently high. From these two bounds, the proof follows.
\end{proof}

\vspace{0.5cm}

\begin{proof}[Proof (Theorem \ref{th:linear_lbound_sparse})]
We now show that when $\Param^0$ is chosen at random from $\mathcal{A}^{(S)}$, the right hand side of \eqref{eq:bound_linear}  reduces to the
right hand side of \eqref{eq:lower_bound_linear_model} in Theorem \ref{th:linear_lbound_sparse}.

Starting from the bound of Lemma \ref{th:mut_inf_bound_linear}, we divide both terms in the numerator and the denominator by $p$. Notice that we can ignore the term $2/p$ in the numerator when $p \rightarrow \infty$.

Recall that $\Param^0$ is built from the adjacency
matrix of a regular graph chosen uniformly
at random and whose entries have had their sign
flipped with probability $1/2$. Therefore,
since $M(\Param^0)$ is the sign-support
of $\Param^0$, we have $H(M(\Param^0)) = \log(|\mathcal{M}|)$. Hence, we can write
$p^{-1}(2H(M(\Param^0)) - \log(|\mathcal{M}|)) = p^{-1} \log(|\mathcal{M}|)$.
In addition, $|\mathcal{M}| =  2^{p k /2} |\mathcal{R}|$, where $|\mathcal{R}|$ is the number of regular graphs of degree $k$ on $p$ nodes and 
$2^{p k/2}$ accounts for the sign flips in the non-zero non-diagonal entries \footnote{Notice that diagonal entries are
constant and equal to $\gamma + 2\param_{\min} \sqrt{k -1}$.}. From \cite{bresler2008reconstruction}, we know that $\log(|\mathcal{R}|) \geq C p k \log(2p/k)$ for small enough constant $C$. And therefore,
$\log(|\mathcal{M}|) / p \geq (k/2) \log(2) + C k \log(2p/k)\geq  C' k \log(2p/k)$ for all $p$ large enough and small enough $C'$.

Lemma \ref{th:random_matrix_calc_sparse} gives an upper bound on the denominator when $p \rightarrow \infty$.

To finish the proof of Theorem \ref{th:linear_lbound_sparse}, we show that $\lim_{p \rightarrow \infty} I(x(0);\Param^0)/p \leq 1$. This finishes the proof since, after multiplying by a small enough constant (only dependent on $k$), the bound obtained by replacing the numerator and denominator with the above limiting lower bounds will be valid for all $p$ large enough.

First notice that $h(x(0)) \leq (1/2) \log (2 \pi e)^p |\E (Q^0)|$ and hence,
\begin{align}
&I(x(0);\Param^0) = h(x(0)) - h(x(0)|\Param^0)\\
&\leq \frac{1}{2} \log (2 \pi e)^p |\E (Q^0)| - \E \frac{1}{2} \log (2 \pi e)^p | Q^0| \label{eq:mut_inf_gaus_bound_1},
\end{align}
where $Q^0 = -(1/2) {(\Param^0)}^{-1}$ is the covariance matrix of the stationary process $x(t)$ and $|.|$ denotes the determinant of a matrix. Then we write, $I(x(0);\Param^0) \leq (1/2)\log |\E (-(\beta \Param^0)^{-1})|  + (1/2) \E \log (|-\beta \Param^0|) \leq \frac{1}{2} \text{Tr}\E (-I -(\beta \Param^0)^{-1})  + \frac{1}{2} \E \text{Tr}\{ -I-\beta \Param^0\}$
where $\beta > 0$ is an arbitrary rescaling factor and the last inequality follows from the matrix inequality $\log(I + (.)) \leq \text{Tr}(.)$. From this and equations \eqref{eq:linear_bound_first_term} and \eqref{eq:linear_bound_second_term} it follows that,
\begin{align}
\lim_{p \rightarrow \infty} \frac{1}{p} I(x(0);\Param^0) \leq -1 + (1/2) (\beta' z + \beta'^{-1} G(k,z))
\end{align}
where $z = \rho/ \param_{\min} + 2 \sqrt{k-1}$ and $\beta' = \beta \param_{\min} $. To finish, note that optimizing over $\beta'$ and then over $z$ gives,
\begin{align}
\beta' z + \beta'^{-1} G(k,z) \leq 2 \sqrt{z G(k,z)} \leq \sqrt{\frac{8(k-1)}{k-2}} \leq 4.
\end{align}
\end{proof}
%
%********************************************************
%
\subsection{Proof of Theorem \ref{th:linear_lbound_dense}}
\label{sec:appendix_lbound_dense_proof}
The proof of this theorem follows closely the proof of Theorem \ref{th:linear_lbound_sparse}. Basically, the claim follows by proving that the bound (\ref{eq:lbound_dense}) holds for an $\Param^0$ chosen at random with a distribution supported on $\cA^{(D)}$.

Again, in order to lower bound the sample-complexity for
random matrices, we make use of Lemma \ref{th:mut_inf_bound_linear}.

Now, however, we construct the random matrix $\Param^0$ as follows.
Let $\tilde{\Param^0}$ be a random symmetric matrix with zero-diagonal and with $\{\param_{ij}\}_{i< j}$ i.i.d. random variables where $\prob(\param_{ij}=\param_{\min}) = \prob(\param_{ij}=-\param_{\min})
=1/4$, and $\prob(\param_{ij}=0) = 1/2$. Notice that the second moment of each entry $i \neq j$ is $\E(\Param_{ij}^2) = \param^2_{\min}/2\equiv \alpha$. We then 
define $\Param^0 = -(\gamma + 2 \sqrt{\alpha})I + \tParam^0/\sqrt{p}$ where 
$\gamma=\gamma(\tParam^0)$ is the smallest value that guarantees that $\lambda_{\min}(-\Param) \ge \rho$.

The following Lemma contains a matrix theory calculation that will be later used in this proof when applying Lemma \ref{th:mut_inf_bound_linear}. Recall that we defined $\alpha = \param^2_{\min}/2$.

\begin{lemma} \label{th:random_matrix_calc_dense}
Let $\Param$ be a random matrix defined as above and 
\begin{align}
Q(\param_{\min},\rho)\equiv \lim_{p \rightarrow \infty} \frac{1}{p} \{ \tr \{ \E(-\Param) \}  - \tr \{ (\E(-\Param^{-1}))^{-1} \} \}.
\end{align}
Then, there exists a constant  $C'$ such that
\begin{align}
&Q(A_{\min},\rho)  \leq \min \{\frac{C' \param^2_{\min} }{2\rho} ,\frac{\param_{\min}}{\sqrt{2}} \}.
\end{align}
\end{lemma}

\begin{proof}
Using Wigner's Semicircle law for random symmetric matrices \cite{Zeitouni} and the bound described in \eqref{eq:jensen_trick} it follows that,
\begin{align}
 &\lim_{p \rightarrow \infty} \frac{1}{p} \{ \text{Tr} \{ \E(-\Param) \}  = \rho + 2\sqrt{\alpha} \text{  \;\;\; and \;\;\;  }\\
 &C(\alpha,\rho) \equiv \lim_{p \rightarrow \infty}  \E \{ \frac{1}{p} \text{Tr} \{ (-\Param)^{-1}\} \}  \label{eq:second_term_wigner}\\
&=\frac{-\sqrt{\rho  \left(4 \sqrt{\alpha }+\rho \right)}+2 \sqrt{\alpha }+\rho }{2 \alpha}.
\end{align}
Since $C(\alpha,\rho) = \alpha^{-1/2} C(1,\rho/\sqrt{\alpha})$, we can write 
$\rho +  2\sqrt{\alpha} - ( C(\alpha,\rho) )^{-1} = \sqrt{\alpha} G(\rho/\sqrt{\alpha})$ where $G(x)$ is a strictly decreasing function. Since $\lim_{\rho \rightarrow 0} = \sqrt{\alpha} G(\rho/\sqrt{\alpha}) = \sqrt{\alpha}$ and $\lim_{\rho \rightarrow \infty}  \rho \sqrt{\alpha} G(\rho/\sqrt{\alpha}) = \alpha$ it follows that there is a constant $C'$ independent of $\alpha$ or $\rho$ such that $\sqrt{\alpha} G(\rho/\alpha) \leq \sqrt{\alpha}\min \{1,C' \sqrt{\alpha} / \rho \}$. The result now follows by replacing $\alpha = \Param^2_{\min}/2$.
\end{proof}

\vspace{0.5cm}

\begin{proof}[Proof (Theorem \ref{th:linear_lbound_dense})]
Like in the proof of Theorem \ref{th:linear_lbound_sparse} we start by dividing both numerator and denominator of \eqref{eq:bound_linear} in Lemma 
\ref{th:mut_inf_bound_linear} by $p$. 
By multiplying the resulting expression by an appropriately small constant we can replace the denominator and $\lim_{p \rightarrow \infty} I(x(0);\Param^0)/p$ by their limits when $p \rightarrow \infty$ and get an expression that is still valid for all $p$ large enough.

Let us produce a lower bound for $2 H(M(\Param^0)) - \log(|\mathcal{M}|)$.
First notice that we again have,
$H(M(\Param^0)) = \log(|\mathcal{M}|)$
since every $M(\Param^0)$ is equally
likely. Therefore,
$2 H(M(\Param^0)) - \log(|\mathcal{M}|) = H(M(\Param^0))$.

Since $H(M(\Param^0))/p = \frac{p-1}{2} H(\{1/2,1/4,1/4 \} ) \geq \frac{(p-1)}{4}  \log 2$ \footnote{$H(\{1/2,1/4,1/4 \} )$ is the entropy of the distribution $\{1/2,1/4,1/4\}$.}, and since by Lemma \ref{th:random_matrix_calc_dense} we already know the limiting expression of the denominator, all we have to do is find $\lim_{p \rightarrow \infty} I(x(0);\Param^0)/p$. By an analysis very similar to that in the proof of Theorem \ref{th:linear_lbound_sparse} one can show that
\begin{align}
\lim_{p \rightarrow \infty} \frac{1}{p} I(x(0);\Param^0)  \leq -1 + \sqrt{(z+2) C(1,z)} \leq 1.
\end{align}
where $C(\alpha,\rho)$ was defined in \eqref{eq:second_term_wigner}, which finishes the proof.
\end{proof}

%
%******************************************************
%
\subsection{Proof of Theorem \ref{th:non_linear_lbound}}

The proof consists in evaluating the lower bound in Corollary 
\ref{th:main_simpler_lbound}. We prove the theorem by showing the bound holds
for functions uniformly chosen over a specific subset of $\cA^{(N)}$. Consider the set of functions such that for each possible support of a $p \times p$ matrix with at most $k$ non-zero entries per row there is one and only one function in the family with $JF$ having that support for all $x$. Note that this
implies that, when evaluating
\ref{th:main_simpler_lbound}, here with $\mathcal{M}  = \cA^{(N)}$, we have $\log(|\mathcal{M}|) = H(M(\Param^0))$. Hence
$2 H(M(\Param^0)) - \log(|\mathcal{M}|) = H(M(\Param^0))$.

Now notice that $\E_{x(0)} \Var_{x(0)|\Param^0} F(x(0);\Param^0) \leq \E (||F(x(0);\Param^0)||^2)$. Secondly notice that, if $x$ and $x'$ only differ on the $j^{th}$ component and $(JF)_{ij} \neq 0$ then $|F_i(x;\Param^0)| \leq |F_i(x';\Param^0)| + D||x'-x||$. Since $JF$ has at most $k$ non-zero entries per row, we get that for any $x$ and $x'$, $|F_i(x;\Param^0)| \leq |F_i(x';\Param^0)| + kD ||x'-x||$. If $x = x(0)$ and $x' = \E_{x(0)|\Param^0}(x(0)|\Param^0)$ 
then squaring the previous expression and taking expectations gives us $\E_{x(0)|\Param^0}(F_i(x;\Param^0)^2|\Param^0) \leq 2 F_i(x';\Param^0)^2 + 2k^2 D^2 B$. From this we get that $\E (||F(x(0);\Param^0)||^2)/p \leq C +  2k^2 D^2 B $ where $C$ is a constant independent of $\Param^0$.
For this sub family of functions we have $H(M(\Param^0)) \geq pk \log(p/k)$ (see \cite{bresler2008reconstruction}). By \eqref{eq:mut_inf_gaus_bound_1}, we know that $I(x(0);\Param^0) \leq (1/2) \log((2 \pi e)^p |\E Q^0| ) - (1/2) \E \log((2 \pi e)^p | Q^0| )$. The first term, which is the entropy of a $p$-dimensional Gaussian with covariance matrix $\E\{Q^0\}$, can be upper bounded by the sum of the entropy of its individual components, which have variance upper bounded by $B$. Finally, since $\lambda_{\min}(Q^0) \geq L$, we have $\log | Q^0| \geq p \log L$ and therefore $I(x(0);\Param^0) \leq p/2 \log B/L$, which completes the proof.
\endproof

\end{document}